\documentclass[preprint,3p, authoryear]{elsarticle} 
\usepackage[utf8x]{inputenc} 

\usepackage[a4paper,colorlinks,breaklinks,unicode]{hyperref}

\usepackage{natbib}

\usepackage{amsmath, amsthm, amssymb}
\usepackage{longtable}
\usepackage{rotating}
\usepackage{xcolor}

\usepackage{tikz}
\usetikzlibrary{automata,positioning}
\usetikzlibrary{shapes.multipart}
\usetikzlibrary{arrows.meta}

\usepackage{subcaption}
\usepackage{multirow}
\usepackage{algorithm}
\usepackage{algorithmic}
\usepackage{rotating}
\usepackage{mathtools}
\usepackage{tabularx}
\usepackage{booktabs}
\usepackage[detect-all]{siunitx}
\usepackage{lscape}
\usepackage[title]{appendix}

\graphicspath{{figures/}}

\newtheorem{thm}{Theorem}[section]
\newtheorem{prop}[thm]{Proposition}
\newproof{pf}{Proof}

\newcommand{\comment}[1]{}
\newcommand{\bftab}{\fontseries{b}\selectfont}

\begin{document}
	
	\title{An Exact Cutting Plane Algorithm to Solve the Selective Graph Coloring Problem in Perfect Graphs \tnoteref{thanks}}
	
	\tnotetext[thanks]{This study is supported by Bo\u{g}azi\c{c}i University Research Fund (grant 11765); and T. Ekim is supported by Turkish Academy of Sciences GEBIP award.}

	\author[1]{Oylum \c{S}eker \corref{cor1}}
	\ead{oylum.seker@utoronto.ca}
	
	\author[2]{T{\i}naz Ekim}
	\ead{tinaz.ekim@boun.edu.tr}
		
	\author[2]{Z. Caner Ta\c{s}k{\i}n}	
	\ead{caner.taskin@boun.edu.tr}
	
	\cortext[cor1]{Corresponding author}
	
	\address[1]{Department of Mechanical and Industrial Engineering, University of Toronto, Toronto, ON, M5S3G8, Canada}
	\address[2]{Department of Industrial Engineering, Bo\u{g}azi\c{c}i University, 34342, Bebek, Istanbul, Turkey}
	
	\begin{abstract}
		We consider the selective graph coloring problem, which is a generalization of the classical graph coloring problem. Given a graph together with a partition of its vertex set into clusters, we want to choose exactly one vertex per cluster so that the number of colors needed to color the selected set of vertices is minimized. This problem is known to be NP-hard. In this study, we focus on an exact cutting plane algorithm for selective graph coloring in perfect graphs. Since there exists no suite of perfect graph instances to the best of our knowledge, we also propose an algorithm to randomly (but not uniformly) generate perfect graphs, and provide a large collection of instances available online. We conduct computational experiments to test our method on graphs with varying size and densities, and compare our results with a state-of-the-art algorithm from the literature and with solving an integer programming formulation of the problem by CPLEX. Our experiments demonstrate that our solution strategy significantly improves the solvability of the problem.

	\end{abstract}
	
	
	\maketitle

\newpage	

\section{Introduction}
\label{sec:intro}

Graph coloring is an assignment of labels, or ``colors" as typically called, to the vertices of a graph in such a way that no pair of vertices that are linked by an edge receives the same color. The general structure of problems that are modeled as a graph coloring problem consists of a set of entities and possible incompatibilities between them, represented by vertices and edges, respectively. A proper color assignment to vertices of a graph means to divide the corresponding set of entities into distinct groups so that incompatible entities do not belong to the same group. Graph coloring arises in a variety of practical areas including scheduling \citep{marx2004graph}, frequency assignment \citep{hale1980frequency}, sudoku puzzles \citep{lewis2015guide}, and more.

As an example to a graph coloring application, suppose that we have a set of tasks, which take place in predetermined and possibly overlapping time slots, for which we want to assign a minimum number of workers. In order to express the problem in the domain of graphs, we can represent each task with a vertex, and link a pair of vertices by an edge if the two tasks take place in the same time slot. If we color the vertices of such a graph, each color corresponds to a distinct worker, and the least number of colors yields the minimum number of workers necessary to carry out all the tasks. 

We can consider a more flexible adaptation of the graph coloring application mentioned above into a ``selective" framework. Suppose that the given set of tasks are divided into distinct groups, or categories, such that it is enough to carry out only one representative task per group. The goal is to perform a sufficient set of tasks with a minimum number of workers. Representing each task with a vertex, task categories with a partition of the vertex set, and pairs of tasks occurring in same time slot with edges, the aim is to select exactly one task from each group in such a way that the number of workers required to complete the selected set of tasks is minimized. When we pick one vertex per group and color the selected ones, each color will correspond to a distinct worker, and the set of vertices having the same color stands for the set of tasks assigned to that particular worker. This new problem is an example to the selective graph coloring problem, and it inherently has two layers; the selection, and the coloring of it.  

As the example graph coloring application and its adaptation into a selective scheme reveal, the \textit{selective graph coloring} problem ({\sc Sel-Col}), which is alternatively referred to as \textit{partition coloring} in the literature, generalizes the classical graph coloring problem. Given a graph together with a partition of its vertex set into clusters, the aim in {\sc Sel-Col} is to choose exactly one vertex from each cluster in such a way that, among all possible vertex selections, the resulting number of colors required to color the selected set of vertices is minimized. When each cluster is comprised of a single vertex, every vertex of the input graph needs to be selected and colored, which makes the classical graph coloring problem a special case of {\sc Sel-Col}.
It has emerged as a model to select routes and assign proper wavelengths in the second step of a two-phase solution approach for the routing and wavelength assignment problem (RWA) in optical networks \citep{li2000partition}.

{\sc Sel-Col} is known to be an NP-hard problem, which follows from the fact that it is a generalized version of the classical graph coloring problem that is NP-hard. Moreover, it remains NP-hard in many special graph families including various subclasses of perfect graphs \citep{demange2015some}, and hence in the general class of perfect graphs. 
There are two tiers in the difficulty of {\sc Sel-Col};  it may be due to an exponential number of possible selections and/or due to the hardness of optimally coloring a given selection \citep{demange2015some}. Hence, rendering one aspect of the problem somehow easy does not  necessarily rid us of the overall difficulty.

To the best of our knowledge, there exist three studies in the literature that concentrate on exact solution methods for {\sc Sel-Col}. The study by \cite{frota2010branch} introduces an integer programming model and a branch-and-cut algorithm for the partition coloring problem.
\cite{hoshino2011branch} propose another integer programming model and a branch-and-price algorithm, which is shown to demonstrate superior performance to the method by \cite{frota2010branch}.
Finally, a recent study by \cite{furini2017exact} proposes a new formulation with an exponential number of variables and designs a branch-and-price algorithm, which improves on the previous exact approaches from the literature. 
In our previous work \citep{seker2019decomposition}, we investigated {\sc Sel-Col} in certain subclasses of perfect graphs, and proposed efficient exact solution algorithms that exploit special characteristics of the graph families under consideration. In this paper, we generalize our earlier approach for {\sc Sel-Col} to the general class of perfect graphs.

There exist various applications that motivate the study of {\sc Sel-Col} in perfect graphs. 
In many real life problems that can be modelled as {\sc Sel-Col}, the application domain yields host graphs that admit characteristics of certain perfect graph families. 
Examples include timetabling problems calling for interval graphs, particular cases of multiple stacks travelling salesperson problem that bring in permutation graphs, quality test scheduling problems yielding linear interval graphs, and more \citep{demange2015some}.
In this regard, {\sc Sel-Col} in the class of perfect graphs serves to consolidate many models sharing a common domain.

The importance of perfect graphs is not confined to {\sc Sel-Col} applications. The class of perfect graphs has led to a key area of interest in graph theory due to the numerous connections it has to a wide range of fields including linear programming and computational complexity.  
It has great significance for several reasons. First, many problems that are NP-hard in general, including the minimum coloring and maximum clique problems, become polynomially solvable when restricted to the class of perfect graphs \citep{grotschel1984polynomial}.
Moreover, for many subclasses of perfect graphs, there exist coloring and clique algorithms that are not only polynomial-time but also of combinatorial nature \citep{golumbic2004algorithmic}. These subclasses, such as chordal graphs, permutation graphs, and interval graphs have additional importance as they naturally arise in various real life applications like perfect phylogeny, DNA sequencing, timetabling, and flight altitude assignment \citep{brandstadt1999graph, spinrad2003efficient, golumbic2004algorithmic}. In this respect, perfect graphs form an umbrella class that unifies the results relating to the complexity of important problems in various graph classes.

The polynomial-time algorithms that can solve some of the aforementioned problems in the general class of perfect graphs, such as the maximum clique problem, are not purely combinatorial; they make use of semidefinite programming models. 
Even though these methods are polynomial-time in theory, they are known to demonstrate poor performance in practice \citep{grotschel1984polynomial}. In order to observe how the performance of such algorithms manifests in practice, it is important to have a collection of perfect graph instances or a method to generate them.

Generation of perfect graphs in its general form, rather than from subclasses of it, has a considerable potential to contribute to the literature by providing a means to test the algorithms specifically designed for perfect graphs. 
To the best of our knowledge, a method to generate general perfect graphs has not been proposed before.
Even though there exists a polynomial-time recognition algorithm for perfect graphs \citep{chudnovsky2005recognizing}, generating a random graph and testing for perfectness may not be a viable course of action to obtain perfect graph instances, because this recognition algorithm is not practical even for small graphs, as pointed out in \citep{yildirim2006extracting}.
Alternatively, one may resort to producing instances from certain subclasses of perfect graphs. 
For instance, two such subclasses for which random generation algorithms are available are chordal graphs \citep{andreou2005generating, markenzon2008two, seker2017linear} and generalized split graphs \citep{mcdiarmid2016random, seker2019decomposition}. 
However, as \cite{yildirim2006extracting} note, this approach would be fairly restrictive in nature since there are at least 120 known subclasses of perfect graphs \citep{hougardy2006classes}.

In this study, we present a cutting plane algorithm for {\sc Sel-Col} in perfect graphs, which is generalization of our previous work \citep{seker2019decomposition}, and also propose a perfect graph generation algorithm, which, to the best of our knowledge, is the first one in the literature that is capable of producing instances from the general class of perfect graphs and hence serves as a first step to fill an important gap in the literature.
Using the proposed generator, which does not guarantee that every perfect graph can be generated with positive probability, we produced a large suite of random perfect graph instances with varying size and densities, and made them accessible online.

We test the performance of our solution approach using the problem instances generated, and compare our results to those of an IP formulation and a branch-and-price algorithm by \cite{furini2017exact}.
The results show that our cutting plane algorithm significantly improves the solution performance, and the improvement manifests most noticeably in low-density graphs (see Section \ref{section:compstudy}). Additionally, we compare the performance of our algorithm for general perfect graphs to that of our previous algorithm tailored for three subclasses of perfect graphs, which are permutation, generalized split, and chordal graphs. Our cutting plane algorithm for general perfect graphs results in better performance in permutation graphs, and marked deterioration in the class of chordal graphs. As for generalized split graphs, we observe that, with the algorithm for general perfect graphs, the overall performance is comparable to our specially tailored algorithm in \citep{seker2019decomposition}.

The rest of this article is organized as follows. In Section \ref{prelim}, we provide some preliminary graph-theoretic definitions and information that relate to perfect graphs and {\sc Sel-Col}.
We give an integer programming formulation and describe our cutting plane algorithm in Section \ref{sec:cuttingplane}, and follow it with the review of two existing methods that we employ within our solution framework in Section \ref{sec:submethods}.
In Section \ref{section:pggen}, we introduce our random perfect graph generation method. 
In Section \ref{section:compstudy}, we report the computational results of our cutting plane approach in comparison to those of the integer programming formulation and a state-of-the-art algorithm by \cite{furini2017exact}.
Finally, we conclude our study in Section \ref{section:conclusion} with a brief summary and possible future research directions.


\section{Definitions}
\label{prelim}

A \textit{graph} $G=(V,E)$ is an ordered pair, where $V$ denotes the set of \textit{vertices} (or \textit{nodes}) and $E$ the set of \textit{edges} that are pairs of vertices. %
A pair of vertices in a graph are called \textit{adjacent} if they are linked by an edge. 
A vertex $u$ is a \textit{neighbor} of a vertex $v$ if there exists an edge $\{u,v\}$. 
The \textit{neighborhood} of a vertex $ v $, $ N(v) $, is the set of all vertices that are adjacent to it. 

The \textit{complement} of a graph $G=(V,E)$, shown as $\bar{G}$, is a graph with the same vertex set $V$ where two distinct vertices of $\bar{G}$ are made adjacent if and only if they are not adjacent in $G$. An \textit{induced subgraph} of $G=(V,E)$ is a graph formed by a subset $ V^{'} $ of $ V $ where the set of edges that exist between pairs in $ V^{'}$ in $G$ are all preserved. For a graph $G=(V,E)$ and  $V'\subseteq V$, we denote the subgraph induced by $V'$ by $G[V']$.

A (simple) \textit{cycle} is a sequence of vertices that are consecutively adjacent, which begins and ends at the same vertex and does not repeat any vertices in between. 
When a cycle contains an odd number of vertices, it is called an \textit{odd cycle}. 
A \textit{clique} in a graph is a subset of vertices such that all vertices in the subset are pairwise adjacent. 
A given clique in a graph is called \textit{maximal} if it cannot be extended by incorporating any other vertex, i.e., if it is not contained within another clique. 
The \textit{clique number} of a graph $G$, denoted by $\omega(G)$, is the number of vertices in a largest clique in $G$. 
A set of vertices in a graph form a \textit{stable set}, or equivalently an \textit{independent set}, if no two vertices in the set are adjacent. The size of a largest stable set in a graph $G$ is called the \textit{stability number} and is shown by $\alpha(G)$.

A coloring of a graph is called a (proper) \textit{$k$-coloring} if it uses at most $k$ colors. 
If a $k$-coloring can be assigned to the vertex set of a graph, then that graph is called {$k$-colorable}. The minimum number of colors necessary to color all vertices of a graph $ G $ is called the \textit{chromatic number} of $ G $, and is denoted by $\chi(G)$.
A graph $G$ is $k$-colorable for all $ k \geq \chi(G)$.

\begin{figure}[!h]
	\centering
	\scalebox{0.6}[0.6]{ 
		\begin{tikzpicture}[main_node/.style={circle,fill=white!80,draw,inner sep=0pt, minimum size=18pt}, line width=0.75pt]	
		\node[main_node] (v1) at (-4,2) {1};
		\node[main_node] (v2) at (-1,2) {2};
		\node[main_node] (v3) at (-1,-1) {3};
		\node[main_node] (v4) at (-4,-1) {4};
		\node[main_node] (v5) at (-3.25,1.25) {5};
		\node[main_node] (v6) at (-1.75,1.25) {6};
		\node[main_node] (v7) at (-1.75,-0.25) {7};
		\node[main_node] (v8) at (-3.25,-0.25) {8};
		\draw (v1) -- (v2) -- (v3) -- (v4) -- (v1) -- (v5) -- (v6) -- (v2);
		\draw (v6) -- (v7) -- (v3);
		\draw (v7) -- (v8) -- (v4);
		\draw (v5) -- (v8);
		\draw[dashed, draw=red, rotate around={38:(-3.75,1.75)}] (-3.75,1.75) ellipse (0.55 and 1.25);
		\draw[dashed, draw=red, rotate around={-38:(-1.25,1.75)}]  (-1.25,1.75) ellipse (0.55 and 1.25);
		\draw[dashed, draw=red, rotate around={38:(-1.25,-0.75)}]  (-1.25,-0.75) ellipse (0.55 and 1.25);
		\draw[dashed, draw=red, rotate around={-38:(-3.75,-0.75)}]  (-3.75,-0.75) ellipse (0.55 and 1.25);
		\node at (-5,2.5) {\textcolor{red}{$V_1$}};
		\node at (0,2.5) {\textcolor{red}{$V_2$}};
		\node at (-5,-1.5) {\textcolor{red}{$V_3$}};
		\node at (0,-1.5) {\textcolor{red}{$V_4$}};
		\node[main_node,fill=gray!40] (v9) at (4,2) {1};
		\node[main_node,fill=gray!100] (v10) at (7,2) {2};
		\node[main_node,fill=gray!40] (v11) at (7,-1) {3};
		\node[main_node,fill=gray!100] (v12) at (4,-1) {4};
		\node[main_node] (v13) at (4.75,1.25) {5};
		\node[main_node] (v14) at (6.25,1.25) {6};
		\node[main_node] (v15) at (6.25,-0.25) {7};
		\node[main_node] (v16) at (4.75,-0.25) {8};
		\draw (v9) -- (v10) -- (v11) -- (v12) -- (v9) -- (v13) -- (v14) -- (v10);
		\draw (v14) -- (v15) -- (v11);
		\draw (v15) -- (v16) -- (v12);
		\draw (v13) -- (v16);
		\draw[dashed, draw=red, rotate around={38:(4.25,1.75)}] (4.25,1.75) ellipse (0.55 and 1.25);
		\draw[dashed, draw=red, rotate around={-38:(6.75,1.75)}]  (6.75,1.75) ellipse (0.55 and 1.25);
		\draw[dashed, draw=red, rotate around={38:(6.75,-0.75)}]  (6.75,-0.75) ellipse (0.55 and 1.25);
		\draw[dashed, draw=red, rotate around={-38:(4.25,-0.75)}]  (4.25,-0.75) ellipse (0.55 and 1.25);
		\node at (3,2.5) {\textcolor{red}{$V_1$}};
		\node at (8,2.5) {\textcolor{red}{$V_2$}};
		\node at (3,-1.5) {\textcolor{red}{$V_3$}};
		\node at (8,-1.5) {\textcolor{red}{$V_4$}};
		\node[main_node,fill=gray!100] (v17) at (12,2) {1};
		\node[main_node] (v18) at (15,2) {2};
		\node[main_node,fill=gray!100] (v19) at (15,-1) {3};
		\node[main_node] (v20) at (12,-1) {4};
		\node[main_node] (v21) at (12.75,1.25) {5};
		\node[main_node,fill=gray!100] (v22) at (14.25,1.25) {6};
		\node[main_node] (v23) at (14.25,-0.25) {7};
		\node[main_node,fill=gray!100] (v24) at (12.75,-0.25) {8};
		\draw (v17) -- (v18) -- (v19) -- (v20) -- (v17) -- (v21) -- (v22) -- (v18);
		\draw (v22) -- (v23) -- (v19);
		\draw (v23) -- (v24) -- (v20);
		\draw (v21) -- (v24);
		\draw[dashed, draw=red, rotate around={38:(12.25,1.75)}] (12.25,1.75) ellipse (0.55 and 1.25);
		\draw[dashed, draw=red, rotate around={-38:(14.75,1.75)}]  (14.75,1.75) ellipse (0.55 and 1.25);
		\draw[dashed, draw=red, rotate around={38:(14.75,-0.75)}]  (14.75,-0.75) ellipse (0.55 and 1.25);
		\draw[dashed, draw=red, rotate around={-38:(12.25,-0.75)}]  (12.25,-0.75) ellipse (0.55 and 1.25);
		\node at (11,2.5) {\textcolor{red}{$V_1$}};
		\node at (16,2.5) {\textcolor{red}{$V_2$}};
		\node at (11,-1.5) {\textcolor{red}{$V_3$}};
		\node at (16,-1.5) {\textcolor{red}{$V_4$}};
		\node at (-2.5,-2.75) {\large{(a)}};
		\node at (5.5,-2.75) {\large{(b)}};	
		\node at (13.5,-2.75) {\large{(c)}};
		%
		\node at (-2.5,3) {};
		\end{tikzpicture}
	}
	\caption{(a) A graph $G$ with a partition of its vertex set into four clusters $\mathcal{V} = \{V_1,\ldots,V_4\}$ shown in dashed ellipses, (b) an optimally colored selection $\{1,2,3,4\}$ in $G$, (c) an optimal selection $\{1,6,3,8\}$ in $G$ with an optimal coloring of it, yielding $\chi_{SEL}(G, \mathcal{V}) = 1$. 	\label{fig:selcol_example}}
	{}
\end{figure}
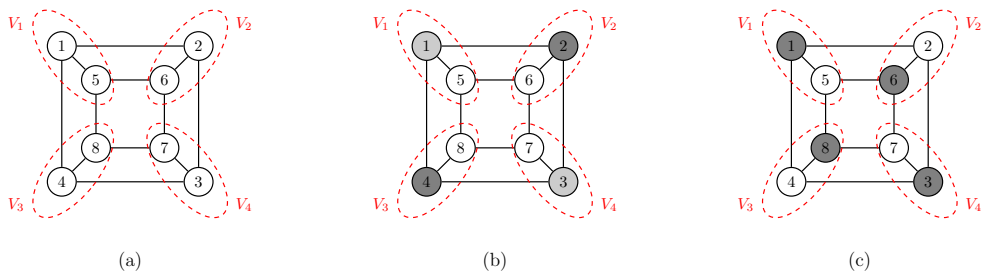

A graph $G$ is called \textit{perfect} if $\chi(G')=\omega(G')$ for every induced subgraph $G'$ of $G$. The \textit{Weak Perfect Graph Theorem} (WPGT) \citep{lovasz1972normal} states that a graph $G$ is perfect if and only if its complement $\bar{G}$ is perfect. The \textit{Strong Perfect Graph Theorem} (SPGT) \citep{chudnovsky2006strong} states that a graph $G$ is perfect if and only if neither $G$ nor $\bar{G}$ contains an odd cycle of length at least five as an induced subgraph.  	
Given a graph $ G = (V, E) $ and a partition $\mathcal{V} = \{V_1,...,V_P\}$ of its vertex set into $P$ clusters, a \textit{selection} is a subset $V'$ of $V$ that comprises exactly one vertex from each one of the clusters in the partition; that is, $V' \subseteq V$ such that $|V' \cap V_p| = 1$ for all $p \in \{1,\ldots, P\}$. 
A \textit{selective $k$-coloring} of $ G $ is defined by a selection $V'$ and a $k$-coloring of $G[V']$. The \textit{selective chromatic number} of a graph $ G $ with vertex partition $\mathcal{V}$, denoted by $\chi_{SEL}(G, \mathcal{V})$, is the smallest integer $k$ for which $G$ admits a selective $k$-coloring \citep{demange2015some}. 
An \textit{optimal selection} is a selection whose optimal coloring yields the selective chromatic number (see Figure \ref{fig:selcol_example}).


\section{Cutting Plane Algorithm for {\sc Sel-Col} in Perfect Graphs}
\label{sec:cuttingplane}

In this section, we start with presenting an integer programming (IP) formulation for {\sc Sel-Col}.
We then describe our cutting plane method that is based on our previous work \citep{seker2019decomposition}, by first providing the model that forms the foundation for it, and afterwards presenting two types of cuts with one being stronger than the other in perfect graphs and explaining the algorithmic procedure.

Suppose that we are given a graph $G=(V,E)$ with ${V = \{1, \ldots ,n\}}$ and a partition $\mathcal{V}$ of its vertex set into $\mathit{P}$ clusters  $V_1,...,V_P$. 
An IP formulation to solve {\sc Sel-Col} can be written as follows:
\begin{subequations}
\label{model:IP}
\begin{alignat}{5}
 &\hspace{-2cm}\textbf{Model 1: }& \qquad &&&\text{min}& \quad &\sum_{k=1}^{\mathit{P}} \ y_{k} \label{IPobj}\\[0.1cm]
 &&&&& \text{s.t.}&  &w_{ik} \ \leq \ y_k \qquad &&  \forall \ i \in V ,\ \ k \in  \{1,...,\mathit{P}\}  \label{IPcons1} \\
&&&&&& &w_{ik} \ + \ w_{jk} \ \leq \ 1 \qquad &&  \forall \ \{i,j\} \in E ,\ \ k \in  \{1,...,\mathit{P}\}  \label{IPcons2} \\
&&&&&& &\sum_{i \in V_p} \sum_{k=1}^\mathit{P} \ w_{ik} \ = \ 1 \qquad && \forall \ p \in \{1,...,\mathit{P}\} \label{IPcons3} \\
&&&&&& &y_k \in \{0,1\} \qquad &&\forall \ k \in  \{1,...,\mathit{P}\} \label{IPcons4} \\
&&&&&& &w_{ik} \in \{0,1\} \qquad && \forall \ i \in V , \ k \in \{1,...,\mathit{P}\} \label{IPcons5} ,
\end{alignat}
\end{subequations}

\noindent where $y_k$ is a binary variable that takes value 1 if color $k$ is used and 0 otherwise, and $w_{ik}$ is another binary variable taking value 1 if vertex $i$ is selected and gets color $k$ and 0 otherwise.

Model 1 takes the number of available colors as $P$, because the size of a selection is $P$ and in the worst case each vertex in the selection takes a distinct color. One should note that a feasible $c$-coloring of a selection can choose any size-$c$ subset of the available $P$ colors. Moreover, a feasible $c$-coloring of a selection has $c!$ equivalent alternatives in the solution space that are obtained by simply permuting the indices of the $c$ colors used. In order to reduce the inherent symmetry in this formulation, we add the constraint set (\ref{symcons}) to Model 1 (similar to the symmetry breaking constraints in \citep{sherali2001improving}). This way, the program is forced to use the colors in increasing order of their indices and clone solutions resulting from alternative combinations of the available $P$ colors are discarded from the solution space. 
\begin{alignat}{1}
y_k  \ \leq \ y_{k-1} \qquad & \forall \ k \in  \{2,...,\mathit{P}\} \label{symcons}
\end{alignat}

Model 1 contains $O(|V| \times P)$ binary variables and $O(|E| \times P)$ constraints. Since it is an integer programming formulation, its solution time and memory requirement may rise exponentially with the increase in the size of the input.

An alternative formulation for {\sc Sel-Col}, which constitutes the basis of our cutting plane algorithm, can be written as follows \citep{seker2019decomposition}: 
\begin{subequations}
\label{model:decomp}
\begin{alignat}{5}
&\hspace{-2cm}\textbf{Model 2: }& \qquad &&&\text{min}& \quad &t \label{obj} \\[0.1cm] 
&&&&& \text{s.t.} & &\sum_{i \in V_p} \ x_i \ = \ 1 \qquad && \forall \ p \in \{1,...,\mathit{P}\}  \label{cons1}\\
&&&&&& &t \geq \chi(G[x])  \label{cons2}\\
&&&&&& &t \geq 0  \label{cons3} \\
&&&&&& &x_i \in \{0,1\} \qquad && \forall \ i \in V \label{cons4},
\end{alignat}
\end{subequations}

\noindent where $x_i$ is a binary variable that is assigned value 1 if vertex $i$ is selected and 0 otherwise. 
$ G[x] $ denotes the graph induced by the selection defined by the variable vector $x=(x_1,\ldots,x_n)$, and the nonnegative variable $t$ is an estimate of the number of colors needed. 

The requirement that exactly one vertex is selected from each of the $P$ clusters is met by constraint set (\ref{cons1}). 
The nonnegative variable $t$ is forced to be at least equal to the chromatic number of the selection given by the variable vector $x$ through constraint set (\ref{cons2}). 
Since the objective is to minimize variable $t$, the optimal objective value of this model will be equal to the selective chromatic number $\chi_{SEL}(G,\mathcal{V})$ of the input graph.
However, enforcing $t$ to be equal to the selective chromatic number is not achieved with linear expressions in the current form of the model.  We need to replace (\ref{cons2}) with a set of linear inequalities that will perform the coloring task. Instead of embedding these inequalities to the model all at once, we are going to generate and incorporate them to the model as needed. In order to be able to do this, we decompose the problem into two parts, and deal with the selection task in one part and the coloring of the given selection in the other. 

We first construct our initial master problem by relaxing the constraint set (\ref{cons2}) and obtaining a linear model that merely yields a feasible vertex selection for $G$. At the beginning, there is no connection between variable $t$ and vector $x$. The link in between is established throughout the iterations. At each step, we start with solving the master problem to optimality and obtain a vertex selection. We then feed this selection to a subproblem where the chromatic number of the graph induced by the given selection is computed. If the chromatic number found by the subproblem is higher than the optimal objective value of the master problem, then it means that the current state of the master problem does not fully incorporate the set of constraints that can correctly estimate the selective chromatic number of the input graph. In this case, we add a constraint to the master problem, which ensures that the $t$-value takes a value at least as large as the chromatic number of the graph induced by the current selection, as long as the same set of vertices is selected. A diagram demonstrating the way our cutting plane algorithm works is provided in Figure \ref{fig:decomp_diagram}. 

\begin{figure}[!h]
  \centering
    \resizebox{0.65\textwidth}{!}{
    	\begin{tikzpicture}[rect/.style={rectangle,fill=white!80,inner sep=6pt, rounded corners, line width=1.2pt, minimum width = 10pt, minimum size = 30pt}]		
    	\hspace{-1cm}
    		\node [minimum width=0.15cm] at (-2.3,2.4) {\bf MP};
    		\node[rect, draw] (v1) at (0.5,2.4) {\parbox{4cm}{\centering { Select vertices \\ \vspace{0.15cm} Guess \# colors as $\boldsymbol{t^{(j)}}$}}};
    		\node [minimum width=0.15cm] at (12.2,2.4) {\bf SP};
    		\node[rect, draw] (v2) at (9.5,2.4) {\parbox{4cm}{\centering {Color selected vertices \\\vspace{0.15cm} Min \# colors used $\boldsymbol{z_{\text{sp}}^{(j)}}$}}};
    		\draw[-{Latex[length=3.5mm, width=2mm]}, bend left=19] (v1.north) to node[swap, anchor=center, below][yshift=0.45cm] {\parbox{2.5cm}{\centering { Master solution \\$\boldsymbol{t^{(j)}}$, $\boldsymbol{x^{(j)}}$}}}(v2.north); 
    		\draw[-{Latex[length=3.5mm, width=2mm]}, bend left=19] (v2.south) to node[swap, anchor=center,below][yshift=0.6cm] {\parbox{2.5cm}{\centering { If {$\boldsymbol{z_{\text{sp}}^{(j)}}$} $\boldsymbol{>}$ {$\boldsymbol{t^{(j)}}$}, \\ add cut to MP}}} (v1.south); 
    	\end{tikzpicture}
	}
    \caption{A diagram showing how the cutting plane algorithm operates}
    \label{fig:decomp_diagram}
\end{figure}
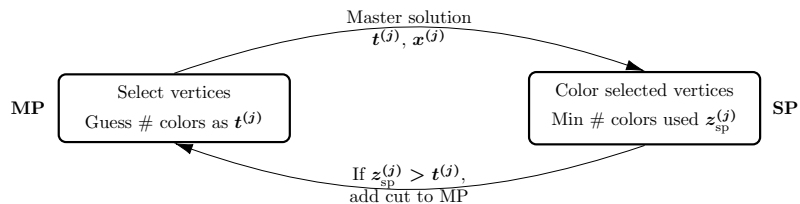

One cut that can be incorporated into our cutting plane framework can be expressed as follows:
 
\begin{alignat}{1}
	& t \ \geq \ \chi(G[x^{(j)}]) \ - \ \sum_{\{i \in V | x^{(j)}_i = 1\}} (1-x_i), \label{cut1}
\end{alignat}
where $G[x^{(j)}]$ denotes the graph induced by the selection found at iteration $j$ given by the variable vector $x^{(j)}$, and $\chi(G[x^{(j)}])$ the chromatic number of this induced subgraph.

\medskip
The inequality (\ref{cut1}) utilizes the fact that the chromatic number of a graph induced by a selection can decrease by at most one for each vertex change. Let us investigate the set of values that the right hand side of constraint (\ref{cut1}) can attain. Firstly, the rightmost summation in inequality (\ref{cut1}) becomes zero only if we choose the exact same set of vertices as in $x^{(j)}$, that is, when $x=x^{(j)}$ holds. Secondly, each time a vertex from the selection given by $x^{(j)}$ is switched to some other vertex, this term increases by one. Considering the term on the right hand side as a whole, the lower bound on $t$ is reduced by one for each vertex we change in the selection.

The solution framework introduced above is valid for any input graph; it does not utilize any particular property of the input graph. As mentioned before, the chromatic number of a perfect graph is equal to the size of a maximum clique in it and by definition, the property of being perfect is \textit{hereditary}, i.e., every induced subgraph of a perfect graph is again perfect. 
We can make use of this property and adapt our solution framework to the class of perfect graphs. Each time the subproblem is called, we can equivalently find a maximum clique in it, instead of its chromatic number. We can express this relationship as an inequality as shown in (\ref{cliquecut}), and utilize it within our cutting plane algorithm. 
\begin{align}
&\hspace{25pt} t \geq \sum_{i \in K^{(j)}} x_i  \label{cliquecut}
\end{align}
where $K^{(j)}$ is a maximum clique of $G[x^{(j)}]$.

Given a selection, the right hand side of (\ref{cliquecut}) is equal to the size of a maximum clique in it. So, for a fixed vertex selection, (\ref{cliquecut}) enforces the master program's objective value $t$ to be at least as large as the number of vertices in a maximum clique of it, as intended. Moreover, (\ref{cliquecut}) provides positive lower bounds for other unexplored selections that intersect with the cliques whose cuts have been added before. For each vertex selection that the master problem outputs, we will construct a constraint of type (\ref{cliquecut}) using the maximum clique found in the subproblem, and add it to the master problem. So, at each iteration, the master problem will contain constraints of type (\ref{cliquecut}) that have been generated so far, together with (\ref{cons1}), (\ref{cons3}) and (\ref{cons4}). We prefer constraint (\ref{cliquecut}) over (\ref{cut1}) because it is stronger for perfect graphs, which we show in the following.

\begin{prop} 
	Constraint (\ref{cliquecut}) is stronger than (\ref{cut1}) for perfect graphs. 
\end{prop}
\begin{proof}
	\begin{NoHyper}
		Given a graph $G$ and a partition $\mathcal{V} = \{V_1,...,V_P\}$ of its vertex set, let us first define the two polyhedra $\mathcal{P}_{\ref{cut1}}$ and $\mathcal{P}_{\ref{cliquecut}}$ as follows:
		\begin{align}
		&\hspace{0pt} \mathcal{P}_{\ref{cut1}} \coloneqq \bigg\{x \in [0,1]^{|V|}, \ t \in \mathbb{R}_{\geq 0} \colon \sum_{i \in V_p} x_i = 1 \ \ \ \forall p \in \{1,...,\mathit{P}\}, \nonumber \\  
		&\hspace{50pt} t \geq \chi(G[\hat{x}]) - \sum_{\{i \in V | \hat{x}_i = 1\}} (1-x_i) \ \ \ \forall \hat{x} \in \{0,1\}^{|V|} \text{ s.t. }  \sum_{i \in V_p} \hat{x}_i = 1 \nonumber \\
		&\hspace{50pt} \forall p \in \{1,...,\mathit{P}\} \bigg\} \nonumber \\
		&\hspace{0pt} \mathcal{P}_{\ref{cliquecut}} \coloneqq \bigg\{x \in [0,1]^{|V|}, \ t \in \mathbb{R}_{\geq 0} \colon  \sum_{i \in V_p} x_i = 1 \ \ \ \forall p \in \{1,...,\mathit{P}\}, \nonumber \\ 
		&\hspace{50pt} t \geq \sum_{i \in \hat{K}} x_i \ \ \ \forall \hat{x} \in \{0,1\}^{|V|} \text{ s.t. } \sum_{i \in V_p} \hat{x}_i = 1 \ \ \ \forall p \in \{1,...,\mathit{P}\} \text{ and } \nonumber \\
		&\hspace{50pt} \hat{K} \text{ is a maximum clique of } G[\hat{x}] \bigg\} \nonumber
		\end{align}
		In other words, if we let $\mathcal{P}$ be the linear programming (LP) relaxation of the feasible solution space defined by the constraint set of our initial master problem, $\mathcal{P}_{\ref{cut1}}$ and $\mathcal{P}_{\ref{cliquecut}}$ are constructed by further constraining $\mathcal{P}$ respectively with constraints (\ref{cut1}) and (\ref{cliquecut}) defined for each one of all possible vertex selections. We want to prove that  $\mathcal{P}_{\ref{cliquecut}} \subseteq \mathcal{P}_{\ref{cut1}} $. To do this, we first show that for any  $\{\bar{t}, \bar{x}\} \in \mathcal{P}_{\ref{cliquecut}}$, $\{\bar{t}, \bar{x}\} \in \mathcal{P}_{\ref{cut1}} $ holds. Since vertex selection constraints are common on both $\mathcal{P}_{\ref{cut1}}$ and $\mathcal{P}_{\ref{cliquecut}}$, $\sum_{i \in V_p} \bar{x}_i = 1 \ \ \forall p \in \{1,...,\mathit{P}\}$ holds by construction. Now, let $\hat{x} \in \{0,1\}^{|V|}$ such that $\sum_{i \in V_p} \hat{x}_i = 1 \ \ \forall p \in \{1,...,\mathit{P}\}$ and $\hat{K}$ is a maximum clique of $G[\hat{x}]$. Note that we can write $\sum_{i \in \hat{K}} \bar{x}_i = |\hat{K}| - \sum_{i \in \hat{K}}(1-\bar{x}_i)$. Since $ \hat{K} \subseteq V(G[\hat{x}])$, we have $\sum_{\{i \in V | \hat{x}_i = 1\}} (1-\bar{x}_i) \geq \sum_{i \in \hat{K}}(1-\bar{x}_i) \geq 0$. As $\chi(G[\hat{x}]) = |\hat{K}|$ by the perfectness of $G$, we have $\bar{t} \geq \sum_{i \in \hat{K}} \bar{x}_i = |\hat{K}| - \sum_{i \in \hat{K}}(1-\bar{x}_i) \geq \chi(G[\hat{x}]) - \sum_{\{i \in V | \hat{x}_i = 1\}} (1-\bar{x}_i)$ and hence $\{\bar{t}, \bar{x}\} \in \mathcal{P}_{\ref{cut1}} $.
		Next, we show that this containment can be strict; i.e., there exists a perfect graph $G$ for which at least one point in $\mathcal{P}_{\ref{cut1}}$ is not contained in $\mathcal{P}_{\ref{cliquecut}}$. To this end, consider the graph $ G = (V, E)$ with $\mathcal{V}=\{V_1, V_2, V_3\}$, where $V = \{1, 2, 3, 4\}$, $E = \emptyset$, $V_1=\{1,2\}$, $V_2=\{3\}$, and $V_3=\{4\}$. There are two possible selections for this graph, which are $\hat{x}^{(1)}=(1,0,1,1)$ and $\hat{x}^{(2)}=(0,1,1,1)$. The constraints of type (\eqref{cut1}) associated with selections $\hat{x}^{(1)}$ and $\hat{x}^{(2)}$ are respectively $c_1: t \geq 1-(3-(x_1+x_3+x_4))$ and $c_2: t \geq 1-(3-(x_2+x_3+x_4))$. Now, take the point $(\bar{t}, \bar{x}_1, \dots, \bar{x}_4) = (0.5,0.5,0.5,1,1)$. This point is contained in $\mathcal{P}_{\ref{cut1}}$, because it satisfies the selection constraints as well as $c_1$ and $c_2$. A maximum clique of $G[\hat{x}^{(1)}]$ is $\{3\}$. The corresponding constraint of type (\ref{cliquecut}), $t \geq x_3$ is violated by the given point, as $\bar{t}=0.5$ and $\bar{x}_3=1$.  Hence, $(0.5,0.5,0.5,1,1) \notin \mathcal{P}_{\ref{cliquecut}}$, and $\mathcal{P}_{\ref{cliquecut}} \subset \mathcal{P}_{\ref{cut1}}$. 
	\end{NoHyper}	
\end{proof}

\begin{figure}[!h]
\centering
	\framebox[6.4in]{\begin{minipage}[t]{6.3in}
		\begin{algorithmic}
			\STATE
			\STATE \textbf{Input:} A perfect graph $G=(V,E)$, and a partition $\mathcal{V}$ of $V$ 
			\STATE \textbf{Output:} An optimal selection $x^*$ with $\chi_{SEL}(G, \mathcal{V})=z^*$
			\medskip				
			\STATE $j \gets 0$, $t^{(j)} \gets 0$, $z_{\text{sp}}^{(j)} \gets \infty$ \\ %
			\WHILE{true} 
			\STATE $j \gets j+1$
			\STATE Solve the master problem optimally, find an optimal selection $x^{(j)}$ with optimal objective value $t^{(j)}$
			\STATE Find a maximum clique $K^{(j)}$ of $G[x^{(j)}]$ in the subproblem
			\STATE $z_{\text{sp}}^{(j)} \gets |K^{(j)}|$ 
			\IF{$z_{\text{sp}}^{(j)} > t^{(j)}$} 
			\STATE Add (\ref{cliquecut}) to the master problem
			\ELSE
			\STATE break
			\ENDIF
			\ENDWHILE
			\STATE $x^* \gets x^{(j)}$, $z^* \gets t^{(j)}$
			\RETURN $x^*, z^*$
			\STATE
		\end{algorithmic}
	\end{minipage}}
	\caption{Cutting Plane Algorithm for Perfect Graphs \label{alg:decompositionpg}}
	{}
\end{figure}
Pseudo-code of our cutting plane algorithm for perfect graphs is provided in Figure \ref{alg:decompositionpg}. At each step $j$ of our cutting plane algorithm, the master problem is solved to optimality yielding a selection $x^{(j)}$ with a corresponding optimal objective value $t^{(j)}$, and $G[x^{(j)}]$ is fed to the subproblem. If the objective value of the subproblem, which is the size of a maximum clique $K^{(j)}$ in $G[x^{(j)}]$, turns out to be greater than $t^{(j)}$, we continue iterating because it means the master problem is currently lacking the constraints that will lead to the correct estimate of the optimal value of $t$. Otherwise, the process is terminated, in which case the incumbent solution $x^*$ and the associated objective value $t^*$ are optimal.

We note that in our computational experiments, we implemented the algorithm in Figure \ref{alg:decompositionpg} using the callback mechanism of the solver, as discussed in Section \ref{section:compstudy}.
We also note that when our cutting plane algorithm terminates with an optimum solution, it does not deliver the colors of the vertices in the optimal selection; we only know the size of a maximum clique and the vertices in it.   
In order to find a minimum coloring of a given optimal selection, we can build an IP model and solve it using the graph induced by the optimal selection as input. 
By allowing only as many colors as the maximum clique size, we can significantly reduce the search space, so that any feasible solution to the minimum coloring problem becomes optimal, too.
To facilitate the solution procedure further, we can incorporate additional constraints that fix the colors of the vertices in the maximum clique to distinct ones.


\section{Methods for the Maximum Clique Problem}
\label{sec:submethods}

In this section, we discuss the two methods that we employ in the subproblem of our cutting plane procedure to find maximum cliques.
First, we discuss the details of an approach by \cite{grotschel1984polynomial} that uses semidefinite programming, and then review a more promising combinatorial method by \cite{tomita2010simple}.

\subsection{Solving the Maximum Clique Problem in Perfect Graphs via Semidefinite Programming}

In the class of perfect graphs, the maximum clique problem is polynomial-time solvable via semidefinite programming (SDP) \citep{grotschel1981ellipsoid, grotschel1984polynomial}.
Finding the clique number of a perfect graph necessitates solving an SDP model only once. However, extracting a maximum clique involves solving a series of SDP models on successively smaller graphs for at most $n$ times, where $ n $ is the number of vertices in the input graph. 

In his seminal paper, \cite{lovasz1979shannon} introduced the so-called \textit{theta function} of a graph, also known as \textit{Lov{\'a}sz's theta function}, which is denoted by $\vartheta(G)$ for a given graph $G$, and satisfies 
\begin{alignat}{1}
	& \alpha(G) \ \leq \ \vartheta(G) \ \leq \ \chi(\bar{G}), \nonumber
\end{alignat}
\noindent where $\alpha(G)$ denotes the size of a maximum stable set in $G$ and $\chi(\bar{G})$ denotes the chromatic number of the complement of $G$.
For any graph $G$, the stability number $\alpha(G)$ equals the clique number of its complement $\omega(\bar{G})$, and $\chi(\bar{G})$ is equal to $\omega(\bar{G})$ for perfect graphs. Then, $\vartheta(G) = \omega(\bar{G})$ holds for perfect graphs. In order to find $\omega(G)$ of a perfect graph $G$, we need to use the complement $\bar{G}$ of it, which is again perfect by the WPGT \citep{lovasz1972normal}. The theta function $\vartheta(G)$ can be computed by several equivalent formulations \citep{knuth1994sandwich, grotschel1988geometric}, as noted in \citep{yildirim2006extracting}. We provide one of these formulations in (\ref{SDP_obj})--(\ref{SDP_cons4}), which is an SDP due to \cite{lovasz1979shannon}.

Let us introduce a few notations first. For two matrices $A \in \mathbb{R}^{n \times n} $ and $ B \in \mathbb{R}^{n \times n} $, the trace inner product is denoted by $A \bullet B = \text{trace}(A^TB)
= \text{trace}(BA^T ) = \sum_{i,j} A_{ij}B_{ij}$. A symmetric real matrix $A$ is said to be \textit{positive semidefinite} if $zA^T z \geq 0$ for every $z \in \mathbb{R}^n$. For an $n \times n$ real symmetric matrix $A$, we use $A \succeq 0 $ to indicate that $A$ is positive semidefinite. Finally, we use $ \mathcal{S}^{n \times n} $ to denote the space of $n \times n$ symmetric matrices. 

Consider the following formulation:
\begin{subequations}
\begin{alignat}{3}
	& \text{max} \qquad &&J \bullet X  \label{SDP_obj}\\ 
	&\text{s.t.} &&I \bullet X \ = \ 1 \qquad && \label{SDP_cons1}\\
	&&& X_{ij} = 0 \qquad && \forall \ \{i,j\}\in E \label{SDP_cons2} \\
	&&& X \ \succeq \ 0 \label{SDP_cons3} \\
	&&& X \in \mathcal{S}^{n \times n} \label{SDP_cons4},
\end{alignat}
\end{subequations}

\noindent where $I$ is the identity matrix, $J$ is a matrix of all ones, and $E$ is the edge set of the input graph. 

The SDP model provided in (\ref{SDP_obj})--(\ref{SDP_cons4}) in general gives an upper bound $\vartheta(G)$ on the stability number $\alpha(G)$ of a graph $G$ \citep{lovasz1979shannon}.  
To see this, let us first consider an IP formulation for the maximum stable set problem for a graph $G=(V,E)$ on $n$ vertices, which is $\max_{x \in \{0,1\}^n} \{e^Tx \colon x_i + x_j \leq 1 \ \ \forall \{i,j\} \in E\}$, and let $\hat{x} \in \{0,1\}^n$ be a feasible solution to it. Let us define $\hat{X} = \hat{x}\hat{x}^T / e^T\hat{x}$, where $e$ is a vector of ones. $\hat{X}$ is positive semidefinite by definition, and this shows that constraints (\ref{SDP_cons3}) and (\ref{SDP_cons4}) are satisfied. Since $\hat{x} \in \{0,1\}^n$, each row of $\hat{x}\hat{x}^T$ is either $\hat{x}^T$ or is comprised of zeroes. As $e^T\hat{x}$ is simply the sum of entries of $\hat{x}$, we have $\hat{X} \in [0,1]^{\{n \times n\}}$. Moreover, $ diag(\hat{X}) = \hat{x} / e^T\hat{x}$, where $diag(\hat{X})$ denotes the vector formed by the diagonal entries of $\hat{X}$. This implies that $\text{trace}(\hat{X}) = 1$, which means that $\hat{X}$ satisfies constraint (\ref{SDP_cons1}). Furthermore, for all $\{i,j\} \in E$, we have $\hat{X}_{ij} = 0$ by definition, which shows that constraint (\ref{SDP_cons2}) is satisfied, too. So, $\hat{X}$ is a feasible solution to the SDP model \citep{galli2017lovasz}. Finally, we have that $J \bullet \hat{X} = \frac{\sum_{i,j} \hat{X}_{ij}}{e^T\hat{x}} = \frac{({e^T\hat{x}})^2}{e^T\hat{x}} = {e^T\hat{x}} $, which means that the objective function simply counts the number of vertices contained in $\hat{x}$. Hence, we conclude that for any graph $G=(V,E)$, any feasible solution to the above IP formulation for the maximum stable set problem can be converted to a feasible solution for the SDP model in (\ref{SDP_obj})--(\ref{SDP_cons4}), which has a corresponding objective value equal to the stability number $\alpha(G)$.  Since there may be other feasible solutions with a better objective value, the optimum objective value of this SDP model provides an upper bound on $\alpha(G)$ in general.  

SDP models can be solved in polynomial time up to any given accuracy \citep{grotschel1981ellipsoid, alizadeh1991interior, nesterov1994interior}.
When the input graph $G$ is perfect, the optimal objective value of the SDP model provided in (\ref{SDP_obj})--(\ref{SDP_cons4}) gives the stability number $\alpha(G)$ \citep{grotschel1984polynomial}. However, we cannot directly obtain a maximum stable set itself by solving this model once. \cite{grotschel1984polynomial} propose a method to extract a maximum stable set in perfect graphs by repeatedly computing the stability number in smaller induced subgraphs of the input graph. The main idea of this method is to remove vertices from the input graph until only the vertices of one maximum stable set remains. It works as follows: First, we find the stability number $\alpha(G)$ of the original input perfect graph $G=(V,E)$. Then, we mark all vertices of $G$ unlabeled. At each step, we select an unlabeled vertex $v \in V(G)$ and tentatively remove it from $G$. Note that $G'= G\setminus \{v\}$ is an induced subgraph of $G$, and hence is perfect, too. We then calculate $\alpha(G')$. If $\alpha(G') = \alpha(G)$, we set $G=G'$, because it means that $v$ is not contained in all maximum stable sets of $G$ and its removal will leave at least one maximum stable set intact. If  $\alpha(G') < \alpha(G)$, then it means that $v$ intersects with all of the maximum stable sets in the current graph and cannot be eliminated. Therefore, we label $v$ in this case and keep it in our vertex set. This process continues until there is no unlabeled vertex, in which case the set of all remaining (labeled) vertices form a maximum stable set of the original graph. Since we either label or remove a vertex at each step, each vertex is considered once in this method. It outputs a maximum stable set after $n$ iterations, with $n$ being the number of vertices of the original input graph. It is also possible to find other maximum stable sets, if any, by changing the order of vertices to be considered. 

This method is the first polynomial-time algorithm to find a maximum stable set in perfect graphs. Since we are interested in finding a maximum clique, which corresponds to a maximum stable set in the complement of the graph, we simply give the complement of the original graph as input, which is also a perfect graph by WPGT. At each step of this method, we make use of the SDP model provided in (\ref{SDP_obj})--(\ref{SDP_cons4}) to find the stability number. Note that the input to this method will be a subgraph of a perfect graph induced by a vertex selection. By definition, induced subgraphs of a perfect graph are again perfect. Therefore, we can safely employ this method in the subroutine of our cutting plane procedure.

We made a minor modification to \cite{grotschel1984polynomial}'s algorithm in order to possibly avoid unnecessary computations. As we compute the size of a maximum stable set at the beginning of the algorithm, we continue iterating until the number of labeled vertices in the graph equals maximum stable set size, instead of waiting for all vertices to be considered.   

Although SDP models are polynomial-time solvable (up to any fixed accuracy) in theory, their practical performance is typically not satisfactory, as will be revealed by the results of our computational experiments presented in Section \ref{section:compstudy}. 
We note that \cite{yildirim2006extracting} propose another algorithm to solve the maximum stable set problem in perfect graphs, which makes use of solutions to SDP models throughout.
Although it shows better performance than the one by \cite{grotschel1984polynomial} in several test instances, we decided not to test it within our cutting plane procedure, because the improvement that their algorithm achieves is far from being comparable to what we achieve by using the combinatorial method by \cite{tomita2010simple}, which we discuss in the sequel.

\subsection{A Branch-and-Bound Algorithm for the Maximum Clique Problem}

A comprehensive review on both exact and heuristic algorithms for maximum clique problem by \cite{wu2015review} provides computational performance comparison of ten state-of-the-art exact algorithms on a set of popular DIMACS instances. One of the best-performing algorithms is that of  \cite{tomita2010simple}, which is a branch-and-bound algorithm that the authors call \textit{MCS}. MCS is based on a previous maximum clique algorithm \textit{MCR} by \cite{tomita2007efficient} and shows considerably improved performance compared to the previous with the help of newly introduced techniques that reduce the search space.

MCR \citep{tomita2007efficient} is a branch-and-bound algorithm that begins with a small clique and continues searching for larger and larger cliques until it finds one that can be confirmed to be of maximum size. At every step, it starts from a single vertex and tries to expand it by adding new vertices. In order to avoid unnecessary searching, the algorithm makes use of a greedy coloring of the set $R$ of common neighbors of vertices in the current clique $Q$. Greedy coloring assigns a minimum possible (integer) label to each vertex in $R$, which simply implies that the size of a maximum clique in $R$, $\omega(R)$, can be at most the maximum label used in greedy coloring. Then, current clique $Q$ can be extended by at most $\omega(R)$ vertices. So, if the sum of $|Q|$ and the maximum label from greedy coloring does not exceed the size of a clique of maximum size found so far, $|Q^*|$, then there is no need to continue searching for vertices to be included in $Q$ because it is simply not possible to obtain a larger clique on that branch. 

In the improved maximum clique algorithm MCS \citep{tomita2010simple}, which we utilize in our cutting plane procedure, the authors focus on reducing the search space further by incorporating a recoloring routine. This routine aims to improve the coloring obtained from the greedy coloring procedure by recoloring vertices with the largest color label into a smaller one. 
One should note that MCS is not tailored for perfect graphs; it works on any graph.


\section{Test Bed: A Perfect Graph Generator}
\label{section:pggen}

In order to test the performance of our solution approach, we need random problem instances. A complete problem instance for {\sc Sel-Col} consists of a graph $G=(V,E)$ and a partition $\mathcal{V}$ of its vertex set $V$. In this section, we first introduce an algorithm to randomly generate perfect graphs and then briefly describe a method to produce random vertex set partitions.

To the best of our knowledge, there have been only theoretical studies on the generation of perfect graphs in its general form. In his survey, \cite{chvatal1984notes} raises the question of whether all perfect graphs are constructible from some ``primitive'' perfect graphs using perfection-preserving operations, while exemplifying some classes all elements of which can be set up through this idea. To date, only some partial answers have been given to this question. For instance, \cite{burlet1984polynomial} have proven that all Meyniel graphs are constructible from certain primitive Meyniel graphs by an operation called \textit{amalgam}. Another study by \cite{chudnovsky2013structure} describes the structure of all bull-free perfect graphs, where \textit{bull} is a graph consisting of a triangle and two vertex-disjoint pendant edges. They show that every bull-free perfect graph either belongs to a basic class, or it can be built from smaller bull-free perfect graphs by an operation that preserves the property of being bull-free and perfect. 

The question of whether all perfect graphs can be built from some primitive perfect graphs still remains to be answered, but there are operations proven to preserve perfection that can seemingly serve well to the purpose of generating perfect graphs. Our algorithm, which we call Algorithm PerfectGen, is based on this idea. We note that Algorithm PerfectGen does not guarantee that every perfect graph can be generated with strictly positive probability or that they are produced uniformly at random. We take a diverse set of small-sized perfect graphs and reach an end-graph by combining randomly selected ones via perfection-preserving operations. 
For this purpose, we made use of the set of perfect graphs up to nine vertices, offered by \cite{mckay}. We filtered out the ones that are not connected, and used the remaining collection to build larger perfect graphs.

Algorithm PerfectGen works as follows: We input a desired number of vertices $n$ and a desired edge density $\rho$ to the algorithm. Initially, we randomly choose a perfect graph from collection $ \mathcal{P} $, which is to be extended into a final perfect graph on $n$ vertices. Then, at each step, we first pick a random perfection-preserving operation $\mathit{op}$ among the six such operations we selected from the literature, whose details we are going to provide in the sequel. If the selected operation $op$ necessitates a perfect graph other than the current perfect graph $G$ that is being extended, then we randomly pick a graph $G'$ from $ \mathcal{P} $ and combine $G$ and $G'$ via operation $\mathit{op}$. Otherwise, we simply apply operation $\mathit{op}$ to $G$. This routine continues until $G$ has $n$ vertices in total.

\begin{figure}[!h]
\centering
	\framebox[6.0in]
	{\begin{minipage}[t]{5.9in}
			\begin{algorithmic}
				\STATE
				\STATE \textbf{Input:} An integer $n$, two real numbers $\rho$ and $\epsilon$ between 0 and 1 
				\STATE \textbf{Output:} A perfect graph $G$ on $n$ vertices with (approximate) edge density $\rho$
				\medskip
				\STATE $d \gets 0$		
				\WHILE { $|d - \rho| < \epsilon $ }			
				\STATE Let $G=(V,E)$ be a graph selected randomly from the collection of small-sized perfect graphs $ \mathcal{P}$ such that $|V| \leq n$
				\WHILE{ $|V| < n$ } 
				\STATE Select a perfection-preserving operation $ \mathit{op} $ randomly
				\IF{$ \mathit{op} $ requires another input graph}
				\STATE Select a random graph $G'=(V',E')$ from $ \mathcal{P} $ with $|V'| \leq n-|V|$
				\STATE Attach $G'$ to $G$ via operation $ \mathit{op} $
				\ELSE
				\STATE Modify $G$ with operation $ \mathit{op}  $
				\ENDIF
				\ENDWHILE
				\STATE $ m \gets |E| $, $d \gets \frac{m}{\frac{n(n-1)}{2}}$
				\IF{$ \rho - \epsilon < 1-d < \rho + \epsilon$}
				\STATE $G \gets \bar{G}$, where $\bar{G}$ is the complement of $G$, $d \gets 1-d$
				\ENDIF
				\ENDWHILE
				\RETURN $G=(V,E)$
				\STATE
			\end{algorithmic}
	\end{minipage}}
	\caption{Algorithm PerfectGen \label{alg:pggen}}
	{}
\end{figure}

The first part of the algorithm explained above has no mechanism to control the number of edges in $G$. In fact, we cannot directly control the number of edges, because the change in the number of edges as well as the number of vertices cannot be foreseen before starting to apply the operation. Moreover, the change in the number of edges is not monotonic throughout the iterations in general; i.e., it can increase, decrease (only possible if we take the complement of the graph), or stay the same. Thus, we first build a perfect graph $G$ on $n$ vertices and then check its edge density $d$. If $d$ is within some predetermined $\epsilon$-distance from the desired edge density $\rho$, then we accept $G$ and terminate the algorithm. On the other hand, if we can achieve the desired density by taking the complement of $G$, then we deliver $\bar{G}$ as the output graph. Otherwise, we simply discard $G$ and start to construct a new perfect graph from scratch. When generating our instances, we set the value of $\epsilon$ as 0.025. Pseudo-code of the algorithm is provided in Figure \ref{alg:pggen}. 

We now present the set of six perfection-preserving operations that we have used in Algorithm PerfectGen.

\begin{itemize}
	\item \textit{Clique identification} \citep{berge1973graphs}: \\
	Let  $G_1$, $G_2$ be disjoint graphs, and $K_i$ be a nonempty clique in $G_i$ satisfying $|K_1|=|K_2|$. Define a one-to-one correspondence between vertices of $K_1$ and $K_2$; i.e., choose a bijective map $f:K_1 \rightarrow K_2$. A graph obtained by unifying each vertex $v$ in $K_1$ with vertex $f(v)$ in $K_2$ is said to arise from $G_1$ and $G_2$ by clique identification. A graph $G$ obtained from two perfect graphs via clique identification is perfect.
	
	In Algorithm PerfectGen, we randomly select one vertex from $ G $ and one vertex from $ G' $, and extend each one to a maximal clique, say $ K_1 $ and $ K_2 $. Without loss of generality, say $ |K_1| \leq |K_2| $. We randomly choose $ |K_1| $ vertices from $ K_2 $ and identify them with those in $ K_1 $. The bijection $ f $ to identify those vertices is randomly determined.  
	
	\begin{figure}[H]
		\centering
		\scalebox{0.6}[0.6]{
			\begin{tikzpicture}[main_node/.style={circle,fill=black!70,draw,inner sep=0pt, minimum size=11pt]}]			
			\node[main_node] (v1) at (-4.25,0) {};
			\node[main_node] (v2) at (-3,0) {};
			\node[main_node] (v5) at (-2,1) {};
			\node[main_node] (v3) at (-1,0) {};
			\node[main_node] (v4) at (0.25,0) {};
			\node[main_node] (v8) at (2.5,0) {};
			\node[main_node] (v7) at (3.5,1) {};
			\node[main_node] (v10) at (4.5,0) {};
			\node[main_node] (v9) at (3.5,-1) {};
			\node[main_node] (v6) at (3.5,2) {};
			\node[main_node] (v11) at (8.75,0) {};
			\node[main_node] (v12) at (10,0) {};
			\node[main_node] (v15) at (11,1) {};
			\node[main_node] (v13) at (12,0) {};
			\node[main_node] (v14) at (13.25,0) {};
			\node[main_node] (v17) at (11,-1) {};
			\node[main_node] (v16) at (11,2) {};
			\draw (v1) -- (v2) -- (v3) -- (v4);
			\draw (v2) -- (v5) -- (v3);
			\draw (v6) -- (v7) -- (v8) -- (v9) -- (v10) -- (v7);
			\draw (v8) -- (v10);
			\draw (v11) -- (v12) -- (v13) -- (v14);
			\draw (v13) -- (v15) -- (v16);
			\draw (v15) -- (v12) -- (v17) -- (v13);
			\draw[dashed, draw=red]  (-2,0.3) ellipse (1.5 and 0.95);
			\draw[dashed, draw=red]  (3.5,0.3) ellipse (1.5 and 0.95);
			\node at (-3.5,1) {\textcolor{red}{$K_1$}};
			\node at (2,1) {\textcolor{red}{$K_2$}};
			\node at (-2,-2) {$G_1$};
			\node at (3.5,-2) {$G_2$};
			\node (v18) at (6.25,0) {};
			\node (v19) at (7.25,0) {};
			\draw[bend right,->,line width=1.1pt]  (v18) -- (v19);
			\node at (11,-2) {$G$};
			\end{tikzpicture}
		}
	\caption{Two perfect graphs combined by clique identification operation \label{fig:cliqueident_example}}
	{}
	\end{figure}
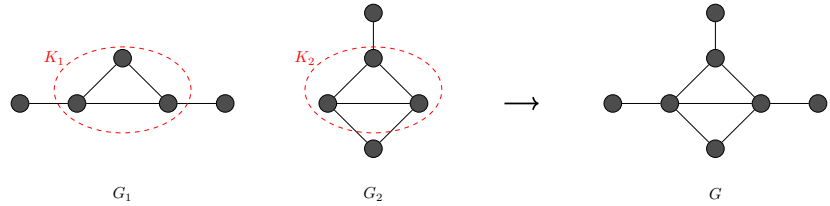
\vspace{-2mm}
	\item \textit{Substitution} \citep{lovasz1972normal}: \\
	Let $G_1$, $G_2$ be disjoint graphs, $v$ be a vertex of $G_1$, and $N$ the set of all neighbors of $v$ in $G_1$. Removing $v$ from $G_1$ and linking each vertex in $G_2$ to those in $N$ results in a graph that arises from $G_1$ and $G_2$ by substitution. If $G_1$ and $G_2$ are perfect, a graph $G$ derived via substitution of the two is perfect too. We note that this operation is also known as \textit{Replication Lemma} in the literature and it played an important role in the proof of the WPGT \citep{lovasz1972normal}. 
	
	Algorithm PerfectGen randomly picks a vertex $v$ from $ G $, and  then substitutes $ v $ with $ G' $ as explained above. Here, $ G $ of Algorithm PerfectGen corresponds to $ G_1 $ above, and similarly $ G' $ corresponds to $ G_2 $. 
	 	
	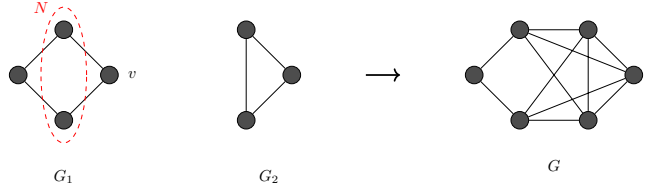
\begin{figure}[H]
		\centering
		\scalebox{0.6}[0.6]{
			\begin{tikzpicture}[main_node/.style={circle,fill=black!70,draw,inner sep=0pt, minimum size=11pt]}]
			\node[main_node] (v1) at (-3.5,0.5) {};
			\node[main_node] (v2) at (-2.5,1.5) {};
			\node[main_node] (v3) at (-1.5,0.5) {};
			\node[main_node] (v4) at (-2.5,-0.5) {};
			\node[main_node] (v5) at (1.5,1.5) {};
			\node[main_node] (v7) at (1.5,-0.5) {};
			\node[main_node] (v6) at (2.5,0.5) {};
			\node[main_node] (v8) at (6.5,0.5) {};
			\node[main_node] (v9) at (7.5,1.5) {};
			\node[main_node] (v10) at (9,1.5) {};
			\node[main_node] (v13) at (7.5,-0.5) {};
			\node[main_node] (v12) at (9,-0.5) {};
			\node[main_node] (v11) at (10,0.5) {};
			\draw (v1) -- (v2) -- (v3) -- (v4) -- (v1);
			\draw (v5) -- (v6) -- (v7) -- (v5);
			\draw (v8) -- (v9) -- (v10) -- (v11) -- (v12) -- (v13) -- (v8);
			\draw (v9) -- (v11) -- (v13) -- (v10) -- (v12) -- (v9);			
			\node (v14) at (4,0.5) {};
			\node (v15) at (5,0.5) {};				
			\draw[bend right,->,line width=1.1pt] (v14) -- (v15);
			\node at (-1,0.5) {$v$};
			\draw[dashed, draw=red] (-2.5,0.5) ellipse (0.5 and 1.5);
			\node at (-2.5,-1.75) {$G_1$};
			\node at (2,-1.75) {$G_2$};
			\node at (-3,2) {\textcolor{red}{$N$}};
			\node at (8.25,-1.5) {$G$};
			\end{tikzpicture}
		}
		\caption{Two perfect graphs combined by substitution operation \label{fig:substitution_example}}
		{}
	\end{figure}	
	\item \textit{``Composition''} \citep{bixby1984composition,william1980combinatorial}: \\
	Let $G_1$, $G_2$ be disjoint graphs each with at least three vertices, $v_i$ be a vertex of $G_i$, $N(v_i)$ the set of all neighbors of $v_i$. The composition of $G_1$ and $G_2$ is obtained from $G_1 \setminus  \{v_1\}$ and $G_2 \setminus \{v_2\}$ by connecting all vertices in $N(v_1)$ to those in $N(v_2)$. A graph obtained from two perfect graphs via composition operation is again perfect. 
	
	In Algorithm PerfectGen, we randomly pick a vertex $v$ from $ G $ and a vertex $ v' $ from $ G' $, and apply the operation as explained above.  
	
	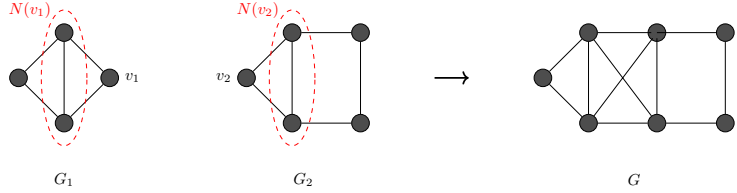
\begin{figure}[H]
		\centering
		\scalebox{0.6}[0.6]{
			\begin{tikzpicture}[main_node/.style={circle,fill=black!70,draw,inner sep=0pt, minimum size=11pt]}]		
			\node[main_node] (v1) at (-1.5,1.5) {};
			\node[main_node] (v2) at (-0.5,0.5) {};
			\node[main_node] (v3) at (-1.5,-0.5) {};
			\node[main_node] (v4) at (-2.5,0.5) {};
			\node[main_node] (v6) at (2.5,0.5) {};
			\node[main_node] (v7) at (3.5,1.5) {};
			\node[main_node] (v5) at (3.5,-0.5) {};
			\node[main_node] (v8) at (5,1.5) {};
			\node[main_node] (v9) at (5,-0.5) {};
			\node[main_node] (v10) at (9,0.5) {};
			\node[main_node] (v11) at (10,1.5) {};
			\node[main_node] (v12) at (10,-0.5) {};
			\node[main_node] at (11.5,1.5) {};
			\node[main_node] (v15) at (11.5,-0.5) {};
			\node[main_node] (v13) at (13,1.5) {};
			\node[main_node] (v14) at (13,-0.5) {};
			\draw (v1) -- (v2) -- (v3) -- (v4) -- (v1) -- (v3);
			\draw (v5) -- (v6) -- (v7) -- (v8) -- (v9) -- (v5) -- (v7);
			\draw (v10) -- (v11) -- (v12) -- (v10);
			\draw (11.5,1.5) node (v16) {} -- (v13) -- (v14) -- (v15) -- (v16) -- (v11) -- (v15) -- (v12) -- (v16);
			\node at (0,0.5) {$v_1$};
			\node at (2,0.5) {$v_2$};
			\node (v17) at (6.5,0.5) {};
			\node (v18) at (7.5,0.5) {};
			\draw[bend right,->,line width=1.1pt] (v17) -- (v18);
			\draw[dashed, draw=red] (-1.5,0.5) ellipse (0.5 and 1.5);
			\draw[dashed, draw=red]  (3.5,0.5) ellipse (0.5 and 1.5);
			\node at (-2.25,2) {\textcolor{red}{$N(v_1)$}};
			\node at (2.75,2) {\textcolor{red}{$N(v_2)$}};
			\node at (-1.5,-1.75) {$G_1$};
			\node at (3.75,-1.75) {$G_2$};
			\node at (11,-1.75) {$G$};
			\end{tikzpicture}
		}
		\caption{Two perfect graphs combined by composition operation \label{fig:composition_example}}
		{}
	\end{figure}	
	\item \textit{Disjoint union}: \\
	Let $G_1$, $G_2$ be two disjoint graphs. The disjoint union of $G_1$ and $G_2$ is simply $G = G_1 \cup G_2$ with $V(G)=V(G_1) \cup V(G_2)$ and $E(G)=E(G_1) \cup E(G_2)$. Disjoint union of two perfect graphs is again perfect (obvious from the definition of perfect graphs).  		
	\item \textit{Join}: \\
	Let $G_1$, $G_2$ be disjoint graphs. The join of $G_1$ and $G_2$, say $G$, is obtained by connecting all vertices in $G_1$ to all those in $G_2$.  A graph obtained from two perfect graphs via join operation is perfect. To show that this operation indeed preserves perfection, assume that $G_1$ and $G_2$ are perfect. Consider $\bar{G}$ which is simply $\bar{G_1} \cup \bar{G_2}$. $G_1$ and $G_2$ being perfect, $\bar{G_1}$ and $\bar{G_2}$ are so, too, by WPGT. As the disjoint union of two perfect graphs is perfect, $\bar{G} = \bar{G_1} \cup \bar{G_2}$ and therefore $G$ is perfect.
	\item \textit{Complement}: \\
	By WPGT, the complement of a perfect graph is again perfect.  
\end{itemize}
\smallskip

The algorithm we designed to generate a random partition of a given vertex set into clusters takes a pair of integers to be respectively the lower and upper bound on the sizes of clusters as input. The first phase of the algorithm initially creates a random ordering $\sigma$ of vertices. Then, at each step, the size $r$ of the cluster under construction is set uniformly random between the lower and upper bounds input to the algorithm, and a separator is placed $r$-many elements ahead of the previous cluster's last vertex in $\sigma$. The set of vertices between two consecutive points the separator is placed serves as one cluster. This procedure continues until all vertices in $V$ belong to some cluster.

All of the perfect graph instances and the associated vertex partitions that we have generated with the presented method can be accessed online at \url{http://www.ie.boun.edu.tr/~taskin/data/pg/}. Our algorithm for random perfect graph generation and the large collection of randomly generated perfect graph instances we provide online serve as a first step to overcome the difficulty of finding perfect graph instances in their general form.


\section{Computational Study}
\label{section:compstudy}

In this section, we present the results of a series of experiments we conducted to evaluate the performance of our cutting plane procedure by comparing it with that of the integer programming formulation Model 1, and the branch-and-price algorithm by \cite{furini2017exact}. 

The algorithms described in the previous section are implemented in C++. We executed the algorithms on a computer with 2.00-GHz Intel Xeon CPU. We used CPLEX version 12.8 in all our experiments, and used the callback mechanism of it,  which enabled us to solve the problem on a single solution tree and generate cuts not only from the optimal solutions to the master problem but also from feasible solutions to it. To solve the SDP formulations, we used MOSEK version 8.1.0.24. The reason for us to select this SDP solver among several others is that MOSEK turned out to be the best-performing one according to the results of benchmark by \cite{mittelmann2018benchmarks} (available at \url{http://plato.asu.edu/ftp/sparse_sdp.html}) conducted on a large set of problem instances, both in terms of solution times and the number of instances that are solved optimally.

We generated random test test instances of with varying size and densities. The number of vertices ($n$) range from 50 to 500. The edge density of a graph is defined as $ \frac{m}{\frac{n(n-1)}{2}} $, where the numerator $m$ denotes the number of edges in the graph, and the denominator is the maximum number edges that it can have. We used four different average edge densities while generating our instances; 0.1, 0.3, 0.5, and 0.7. For each pair of $n$ and average edge density value, we used five random graph instances. 

When an instance could not be solved to optimality by any of the methods we consider, we report the optimality gap percentage, which is calculated as $\frac{UB-LB}{UB}\times 100$ with $UB$ and $LB$ denoting the upper and lower bounds respectively, to give an indication of how far a feasible solution is away from optimal.

We set a time limit of 1200 seconds throughout all the experiments for each one of the methods we experiment with. When an instance could not be solved optimally within the limit, we take the solution time of that instance as 1200 seconds. In our experiments, the B\&P algorithm by \cite{furini2017exact} failed to report optimality gaps for instances that could not be solved optimally within the time limit. Therefore, we use the optimality gap values only when comparing the cutting plane algorithm to the IP formulation.

In our first set of experiments, we test the performance of our cutting plane approach for perfect graphs using the SDP-based method by \cite{grotschel1984polynomial} in the subproblem versus using the maximum clique algorithm MCS by \cite{tomita2010simple}. Table \ref{tab:SDPvsTomita} summarizes the computational results for perfect graph instances with cluster sizes varying between 2 and 5. 
The first three columns of this table provides some information about the instances by respectively listing the number of vertices (``$n$''), average edge density (``Avg density''), and average number of clusters (``Avg \# clust'') across five random instances. The next two sets of columns show the results of our experiments for the two versions of our algorithm for perfect graphs under ``Cutting Plane w/ SDP'' and ``Cutting Plane w/ MCS'' headings, respectively. For the cutting plane method coupled with the SDP-based method of \cite{grotschel1984polynomial}, Table \ref{tab:SDPvsTomita} respectively lists the number of instances that could be optimally solved among five (``\# opt''), average optimality gap percentages over instances that were not optimally solved within the given time limit of 1200 seconds (``Avg \% gap in nonopt''), average solution time in seconds over instances that are optimally solved (``Avg time in opt''), and average solution time over all instances (``Avg overall time'') in columns 4--7. Columns 8--11 contain the same set of results as columns 4--7 for the cutting plane method coupled with MCS. 
Each row of the table reports the average values across runs on five independent instances,
where the bottom row recaps the results by providing the sum for ``\# opt'' column and the averages for the others. 

\medskip

\begin{table}[htbp]
	\caption{Experimental results for perfect graph instances with small clusters to compare the SDP-based method with the MCS of \cite{tomita2010simple}}
	\label{tab:SDPvsTomita}
		\centering
		\resizebox{\textwidth}{!}{
			\begin{tabular}{cc S[table-format=3.1] c S[table-format=2.2] S[table-format=3.2] S[table-format=4.2] c S[table-format=2.2] S S[table-format=4.2]}
				\toprule
				&    &    &  \multicolumn{4}{c}{\textbf{Cutting Plane w/ SDP}} & \multicolumn{4}{c}{\textbf{Cutting Plane w/ MCS}} \\
				\cmidrule(lr){4-7}\cmidrule(lr){8-11}
				\parbox[t]{1.4cm}{\centering $\boldsymbol{n}$}  &  \parbox[t]{1.7cm}{\centering \textbf{Avg\\density}} & \parbox[t]{1.7cm}{\centering \textbf{Avg \\ \# \\ clust}} & \parbox[t]{1.5cm}{\centering \textbf{\#\\opt}} & \parbox[t]{1.7cm}{\centering \textbf{Avg \\ \% gap \\nonopt}} & \parbox[t]{1.8cm}{\centering \textbf{Avg \\ time\\in opt}} & \parbox[t]{1.7cm}{\centering \textbf{Avg \\ overall \\time}} & \parbox[t]{1.5cm}{\centering \textbf{\#\\opt}} & \parbox[t]{1.7cm}{\centering \textbf{Avg \\ \% gap \\nonopt}} & \parbox[t]{1.8cm}{\centering \textbf{Avg \\ time\\in opt}} & \parbox[t]{1.7cm}{\centering \textbf{Avg \\ overall \\ time}} \\				
				\toprule

				\multirow{4}[0]{*}{50} & 0.110 & 14.4  & 5     &       & 3.15  & 3.15  & 5     &       & 0.28  & 0.28 \\
				& 0.293 & 13.8  & 5     &       & 2.21  & 2.21  & 5     &       & 0.20  & 0.20 \\
				& 0.495 & 14.2  & 5     &       & 3.44  & 3.44  & 5     &       & 0.14  & 0.14 \\
				& 0.710 & 14.0  & 5     &       & 6.50  & 6.50  & 5     &       & 0.35  & 0.35 \\
				\midrule
				\multirow{4}[0]{*}{100} & 0.096 & 28.4  & 5     &       & 94.59 & 94.59 & 5     &       & 0.29  & 0.29 \\
				& 0.300 & 29.4  & 5     &       & 88.51 & 88.51 & 5     &       & 0.19  & 0.19 \\
				& 0.488 & 28.0  & 5     &       & 105.39 & 105.39 & 5     &       & 0.34  & 0.34 \\
				& 0.705 & 28.8  & 5     &       & 314.56 & 314.56 & 5     &       & 1.41  & 1.41 \\
				\midrule
				\multirow{4}[0]{*}{150} & 0.098 & 42.6  & 3     & 50.00 & 580.73 & 828.44 & 5     &       & 0.28  & 0.28 \\
				& 0.298 & 42.2  & 5     &       & 720.72 & 720.72 & 5     &       & 0.20  & 0.20 \\
				& 0.498 & 44.2  & 1     & 34.58 & 624.60 & 1084.92 & 5     &       & 0.50  & 0.50 \\
				& 0.693 & 43.4  & 0     & 34.29 &       & 1200.00 & 5     &       & 6.12  & 6.12 \\
				\midrule
				\multirow{4}[0]{*}{200} & 0.107 & 56.8  & 0     & 53.33 &       & 1200.00 & 5     &       & 0.21  & 0.21 \\
				& 0.304 & 57.2  & 0     & 48.33 &       & 1200.00 & 5     &       & 0.20  & 0.20 \\
				& 0.496 & 57.2  & 0     & 48.03 &       & 1200.00 & 5     &       & 0.67  & 0.67 \\
				& 0.703 & 57.6  & 0     & 56.20 &       & 1200.00 & 4     & 9.09  & 232.05 & 425.64 \\
				\midrule
				\multirow{4}[0]{*}{250} & 0.112 & 71.4  & 0     & 77.33 &       & 1200.00 & 5     &       & 0.19  & 0.19 \\
				& 0.304 & 71.2  & 0     & 82.42 &       & 1200.00 & 5     &       & 0.42  & 0.42 \\
				& 0.497 & 72.2  & 0     & 76.96 &       & 1200.00 & 5     &       & 1.89  & 1.89 \\
				& 0.693 & 70.2  & 0     & 70.12 &       & 1200.00 & 2     & 13.33 & 254.25 & 821.70 \\
				\midrule
				\multirow{4}[0]{*}{300} & 0.110 & 86.6  & 0     &   {--}    &       & 1200.00 & 5     &       & 0.23  & 0.23 \\
				& 0.302 & 88.6  & 0     &   {--}    &       & 1200.00 & 5     &       & 1.04  & 1.04 \\
				& 0.506 & 83.4  & 0     & 92.86 &       & 1200.00 & 5     &       & 4.72  & 4.72 \\
				& 0.691 & 85.6  & 0     & 86.41 &       & 1200.00 & 2     & 18.51 & 319.48 & 847.79 \\
				\midrule
				\multirow{4}[0]{*}{350} & 0.117 & 102.6 & 0     &   {--}    &       & 1200.00 & 5     &       & 0.23  & 0.23 \\
				& 0.301 & 100.0 & 0     &   {--}    &       & 1200.00 & 5     &       & 0.51  & 0.51 \\
				& 0.508 & 98.4  & 0     & 96.91 &       & 1200.00 & 4     & 8.33  & 15.55 & 252.44 \\
				& 0.698 & 99.6  & 0     & 96.21 &       & 1200.00 & 0     & 18.72 &       & 1200.00 \\
				\midrule
				\multirow{4}[0]{*}{400} & 0.111 & 114.8 & 0     &   {--}    &       & 1200.00 & 5     &       & 0.34  & 0.34 \\
				& 0.315 & 114.0 & 0     &   {--}    &       & 1200.00 & 5     &       & 0.65  & 0.65 \\
				& 0.502 & 114.2 & 0     &   {--}    &       & 1200.00 & 4     & 8.33  & 25.32 & 260.25 \\
				& 0.692 & 112.8 & 0     & 97.71 &       & 1200.00 & 0     & 24.55 &       & 1200.00 \\
				\midrule
				\multirow{4}[0]{*}{450} & 0.117 & 130.2 & 0     &   {--}    &       & 1200.00 & 5     &       & 0.39  & 0.39 \\
				& 0.309 & 130.4 & 0     &  {--}     &       & 1200.00 & 5     &       & 2.02  & 2.02 \\
				& 0.507 & 128.4 & 0     &   {--}    &       & 1200.00 & 4     & 10.00 & 64.26 & 291.41 \\
				& 0.696 & 125.8 & 0     &   {--}    &       & 1200.00 & 0     & 33.30 &       & 1200.00 \\
				\midrule
				\multirow{4}[0]{*}{500} & 0.117 & 142.8 & 0     &   {--}    &       & 1200.00 & 5     &       & 0.59  & 0.59 \\
				& 0.296 & 143.0 & 0     &   {--}    &       & 1200.00 & 5     &       & 1.05  & 1.05 \\
				& 0.507 & 144.2 & 0     &   {--}    &       & 1200.00 & 1     & 15.22 & 74.45 & 974.89 \\
				& 0.695 & 141.0 & 0     & 98.36 &       & 1200.00 & 0     & 30.15 &       & 1200.00 \\
				\bottomrule \\[-0.3cm]
				&       &       & {\bftab 49} & {\bftab 70.59} & {\bftab 231.31} & {\bftab \phantom{0}951.31} & {\bftab 166} & {\bftab 17.23} & {\bftab 28.08} & {\bftab \phantom{0}217.50} \\
				\bottomrule
	\end{tabular}%
	}
\end{table}%

We observe from the results listed in Table \ref{tab:SDPvsTomita} that solving the subproblem via the MCS algorithm by \cite{tomita2010simple} clearly yields superior results in terms of the number of instances solved to optimality, average optimality gap, and average amount of time spent. 
As $n$ and edge density increase, the performance of both methods deteriorate as expected; however, coupling of the cutting plane method with MCS outperforms the other in every aspect for all the instances. Out of the 200 instances we experiment with, the version that uses MCS could optimally solve 166 of them, whereas the one that solves SDP models could only solve 49 instances to optimality. Moreover, when we use the SDP-based method in the subproblem, we observe that in many instances with 300 or more vertices, the algorithm could not even finish solving the maximum clique problem for the first selection the master problem outputs. In such cases, no optimality gap could be reported, which is revealed by the cells with a ``--" sign in ``Avg \% gap in nonopt'' column for groups of instances for which the number of optimally solved instances shown in the fourth column is zero. In terms of the overall averages shown in the bottom row, when the SDP-based method is used in the subproblem, the average percentage gap and average time spent over all instances are three times higher, and the average time spent in optimally solved instances is seven times higher. 

The results in Table \ref{tab:SDPvsTomita} show that using the SDP-based method in the subproblem of the cutting plane method leads to relatively poor performance in all respects. Therefore, we utilize MCS in the subproblem of our solution procedure for the rest of our computational experiments. In the remaining portion of this section, we present the experimental results of the IP formulation, B\&P algorithm by \cite{furini2017exact}, and our cutting plane method. 

We tested the three methods on 600 test instances in total, 200 for each cluster size range. In order to give an idea on the general performances first, we begin with presenting a brief synopsis of our overall results in Table \ref{tab:perfectsummary}. 
In the first two columns of Table \ref{tab:perfectsummary}, we provide information about the cluster sizes and graph densities by combining them into three groups. The ``small", ``medium", and ``large" cluster sizes respectively correspond to those having 2--5, 4--7, and 6--9 vertices in them. The ``low" and ``high" densities denote the edge densities between 0.1--0.3 and 0.5--0.7, respectively, where the ``all" category contains the entire set of edge densities.
In the next three sets of columns, we report the results of the three methods under ``IP formulation", ``B\&P", and ``Cutting plane" headings. For each method, we provide the number of optimally solved instances (``\# opt"), average optimality gap percentages (``Avg \% gap"), and average solution times in seconds (``Avg time"), except that we do not report the gap values for the B\&P method as mentioned before.
The average gap percentages and solution times in each row incorporate all instances in that category, both the optimally solved ones and the others. 
Each row of the table provides the total or average values across all $n$ values for that category, where the last one shows the total number of optimally solved instances and the overall averages of gaps and solution times.

\begin{table}[htbp]	
	\caption{Summary of experimental results for all perfect graph instances}
	\label{tab:perfectsummary}
	{	\centering
		\resizebox{0.99\textwidth}{!}{
			\begin{tabular}{c c S[table-format=3.0] S[table-format=2.2] S[table-format=3.2] c S[table-format=3.0] S[table-format=3.2] c S[table-format=3.0] S[table-format=2.2] S[table-format=3.2]}
				\toprule
				&    & \multicolumn{3}{c}{\textbf{IP formulation}} & & \multicolumn{2}{c}{\textbf{B\&P}} & & \multicolumn{3}{c}{\textbf{Cutting plane}}\\
				\cmidrule(lr){3-6}\cmidrule(lr){7-9}\cmidrule(lr){10-12}
				\parbox[t]{1.8cm}{\centering \textbf{Sizes of \\clusters}} &\parbox[t]{1.8cm}{\centering \textbf{Density}} & \parbox[t]{1.6cm}{\centering \textbf{\#\\opt}} & \parbox[t]{1.6cm}{\centering \textbf{Avg \\ \% gap} } & \parbox[t]{1.6cm}{\centering \textbf{Avg \\ time}} && \parbox[t]{1.6cm}{\centering \textbf{\#\\opt}} & \parbox[t]{1.6cm}{\centering \textbf{Avg \\ time }} && \parbox[t]{1.6cm}{\centering \textbf{\#\\opt}} & \parbox[t]{1.6cm}{\centering \textbf{Avg \\ \% gap}} & \parbox[t]{1.6cm}{\centering \textbf{Avg \\ time }} \\
				\midrule
				
				\multirow{3}[0]{*}{small} & low & 78 & 19.37 & 372.23 &    & 53  & 731.35 &    & 100 & 0.00 & 0.48 \\
				& high & 48 & 44.03 & 707.97 &    & 88  & 469.88 &    & 66 & 7.26 & 434.51 \\[0.05cm]
				& all & 126 & 31.70 & 540.10 &    & 141 & 600.62 &    & 166 & 3.63 & 217.50 \\
				\midrule
				\multirow{3}[0]{*}{medium} & low & 97 & 2.75 & 175.36 &    & 80 &  339.75 &    & 100 & 0.00 & 0.96 \\
				& high & 73 & 24.17 & 448.30 &    & 92 &  365.42 &    & 88 & 3.57 & 177.81 \\[0.05cm]
				& all & 170 & 13.46 & 311.83 &    & 172  & 352.59 &    & 188 & 1.78 & 89.39 \\
				\midrule
				\multirow{3}[0]{*}{large} & low & 100 & 0.00 & 72.54 &    & 100 & 9.34 &    & 100 & 0.00 & 0.47 \\
				& high & 85 & 15.00 & 302.46 &    & 92  & 331.86 &    & 100 & 0.00 & 8.42 \\[0.05cm]
		        & all & 185 & 7.50 & 187.50 &    & 192 & 170.60 &    & 200 & 0.00 & 4.44\\
				\midrule
				\multicolumn{2}{c}{} & {\bftab 481} & {\bftab 17.55} & {\bftab 346.48} &    & {\bftab 505} & {\bftab 374.60} &    & {\bftab 554} & {\bftab \phantom{0}1.80} & {\bftab 103.77} \\
				\bottomrule				
			\end{tabular}
	}
}
{}
\end{table}%

Firstly, the values in the last row of Table \ref{tab:perfectsummary} indicate that the cutting plane algorithm delivers the best overall performance in terms of the number of optimally solved problems, as well as optimality gap and solution time.
Out of the 600 instances in total, the cutting plane algorithm solves about 92\% of them to optimality, whereas IP and B\&P could solve 80\% and 84\%, respectively. In terms of the average optimality gap, our algorithm yields an order of magnitude better optimality gaps on the average as compared to the IP formulation, and the average solution time is about 30\% and 27\% of those of IP and B\&P, respectively. 
A closer look at the results in terms of density categories reveals that our cutting plane algorithm is successful in solving the low-density instances particularly.
The B\&P algorithm can solve a higher number of high-density instances optimally for clusters with small and medium size than does the IP formulation and the cutting plane method, but it yields larger average solution times than those of the cutting plane method.
Regarding the effect of cluster sizes, increasing the number vertices per cluster leads to improved performance in all three methods, regardless of edge densities. In case of large clusters in particular, the cutting plane algorithm can optimally solve all the instances in a matter of a few seconds on the average.

Having obtained a first impression on the overall results, we next evaluate the performances of the three methods in more detail for each cluster size category separately. We start with the test instances having small clusters, in which perfect graphs are coupled with clusters containing 2 to 5 vertices.
Figure \ref{fig:comp_results_cl_2_5} illustrates a comparison of performances of the three methods. 
The three charts in Figures \ref{subfig_a:perc_opt_cl_2_5}--\ref{subfig_c:time_all_cl_2_5} respectively show the percentage of instances solved optimally within the given time limit (``\% opt"), average solution time of optimally solved instances (``Avg time in opt"), and average solution time over all instances (``Avg time") versus the number of vertices of input graphs ($n$). Each point in these charts give the average value over 20 instances, five for each one of the four different density values, for a given $n$.
In terms of the percentage of optimally solved instances (Figure \ref{subfig_a:perc_opt_cl_2_5}), the cutting plane algorithm outperforms the other two methods except for two values of $n$. 
The solution times of the cutting plane method is almost always the best as compared to those of the other two methods, and the difference is notably high most of the time (Figure \ref{subfig_c:time_all_cl_2_5}). 
Also, the average solution time of the cutting plane method and IP in optimally solved instances does not vary much with increasing $n$ values, whereas those of B\&P exhibit an increasing trend (Figure \ref{subfig_b:time_opt_cl_2_5}). 

\begin{figure}[h]
	\centering
	\begin{subfigure}{.33\textwidth}
		\centering
		\includegraphics[width=.99\linewidth]{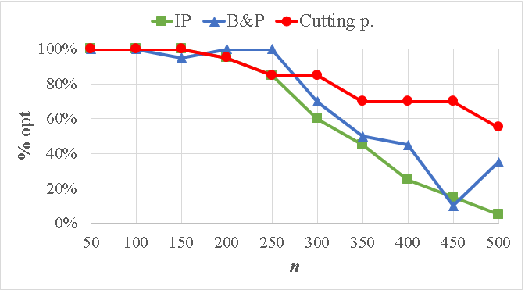}
		\caption{\footnotesize Percentage of optimally solved}
		\label{subfig_a:perc_opt_cl_2_5}
	\end{subfigure}%
	\begin{subfigure}{.33\textwidth}
		\centering
		\includegraphics[width=.99\linewidth]{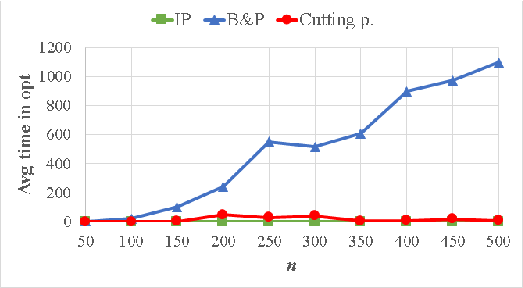}
		\caption{\footnotesize Average time in optimally solved}
		\label{subfig_b:time_opt_cl_2_5}
	\end{subfigure}
	\begin{subfigure}{.33\textwidth}
		\centering
		\includegraphics[width=.99\linewidth]{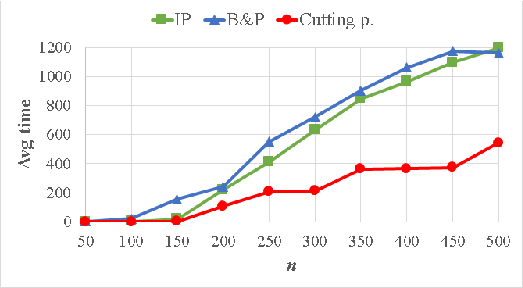}
		\caption{\footnotesize Average time over all instances}
		\label{subfig_c:time_all_cl_2_5}
	\end{subfigure}
	\caption{Results on perfect graph instances with 2-5 vertices in each cluster}
	\label{fig:comp_results_cl_2_5}
\end{figure}

We next provide a more detailed summary of the results for instances with small clusters in Tables \ref{tab:perfectsmall_low} and \ref{tab:perfectsmall_high}. Table \ref{tab:perfectsmall_low} summarizes the results for perfect graph instances having edge density 0.1 and 0.3 (which are referred to as {\it low density}). The first three columns are the same as in Table \ref{tab:SDPvsTomita}. 
The following three sets of columns with headings ``IP formulation'', ``B\&P'', and ``Cutting Plane'' contains the results for the associated algorithms we experiment with. 
The five columns under `IP formulation'' and ``Cutting Plane'' headings report the same set of results respectively for the two methods, which are the number of instances that are solved to optimality among five instances (``\# opt''), average optimality gap percentages over instances that could not be solved optimally within the given time limit of 1200 seconds (``Avg \% gap in nonopt'') and over all instances (``Avg \% gap overall''), average solution time in seconds over instances that are optimally solved (``Avg time in opt'') and over all instances (``Avg time overall''). The rightmost column of the table lists the average percentage of time spent in the subproblem of our cutting plane algorithm across five instances (``Avg \% time in subpr''). 
For B\&P of \cite{furini2017exact}, columns associated with the optimality gap values are excluded.
Each row of this table reports the average values across five independent problem instances. Finally, the bottom row provides the totals for columns containing the number of instances solved optimally (``\# opt''), and the averages for the others.

\begin{sidewaystable}[htbp]
	\centering
	\caption{\label{tab:perfectsmall_low} Experimental results for low-density perfect graph instances with small clusters}
	\resizebox{\textwidth}{!}{
		\begin{tabular}{c c S[table-format=3.1] 
		c S[table-format=3.2] S[table-format=3.2] S[table-format=4.2] S[table-format=4.2]
		c S[table-format=4.2] S[table-format=4.2] 
		c S[table-format=3.2] S[table-format=1.2] S[table-format=1.2] S[table-format=1.2] S[table-format=1.2] }
			\toprule
			&    &    & \multicolumn{5}{c}{\textbf{IP formulation}} & \multicolumn{3}{c}{\textbf{B\&P}} & \multicolumn{6}{c}{\textbf{Cutting Plane}}\\
			\cmidrule(lr){4-8}\cmidrule(lr){9-11}\cmidrule(lr){12-17}
			\parbox[t]{1.3cm}{\centering $\boldsymbol{n}$}  &  \parbox[t]{1.3cm}{\centering \textbf{Avg\\density}} & \parbox[t]{1.3cm}{\centering \textbf{Avg \\ \# \\ clust}} & \parbox[t]{1.3cm}{\centering \textbf{\#\\opt}} & \parbox[t]{1.3cm}{\centering \textbf{Avg \\ \% gap \\in \\ nonopt}} & \parbox[t]{1.3cm}{\centering \textbf{Avg \\ \% gap \\overall}} & \parbox[t]{1.5cm}{\centering \textbf{Avg \\ time \\in opt}} & \parbox[t]{1.5cm}{\centering \textbf{Avg \\ time \\overall}} &
			\parbox[t]{1.3cm}{\centering \textbf{\#\\opt}} & \parbox[t]{1.5cm}{\centering \textbf{Avg \\ time \\in opt}} & \parbox[t]{1.5cm}{\centering \textbf{Avg \\ time \\overall}} &
			\parbox[t]{1.3cm}{\centering \textbf{\#\\opt}} & \parbox[t]{1.3cm}{\centering \textbf{Avg \\ \% gap \\in \\ nonopt}} & \parbox[t]{1.3cm}{\centering \textbf{Avg \\ \% gap \\overall}} & \parbox[t]{1.5cm}{\centering \textbf{Avg \\ time \\in opt}} & \parbox[t]{1.5cm}{\centering \textbf{Avg \\ time \\overall}} & \parbox[t]{1.3cm}{\centering \textbf{Avg \\ \% time \\in \\subpr}}   \\
			\midrule				
			{\multirow{2}[1]{*}{50}} & 0.110 & {14.4} & 
			5  &    & 0.00 & 0.13 & {0.13} & 5  & 0.14 & {0.14} & 5  &    & 0.00 & 0.28 & 0.28 & {45.43} \\
			{} & 0.293 & {13.8} & 5  &  & 0.00 & 0.21 & {0.21} & 5  & 2.02 & {2.02} & 5  &    & 0.00 & 0.20 & 0.20 & {37.60} \\
			\midrule
			{\multirow{2}[0]{*}{100}} & 0.096 & {28.4} & 
			5  &    & 0.00 & 0.90 & {0.90} & 5  & 23.07 & {23.07} & 5  &    & 0.00 & 0.29 & 0.29 & {44.59} \\
			{} & 0.300 & {29.4} & 5  &    & 0.00 & 3.24 & {3.24} & 5   & 32.02 & {32.02} & 5  &    & 0.00 & 0.19 & 0.19 & {45.24} \\
			\midrule
			{\multirow{2}[0]{*}{150}} & 0.098 & {42.6} & 5  &    & 0.00 & 2.97 & {2.97} & 5   & 146.58 & {146.58} & 5  &    & 0.00 & 0.28 & 0.28 & {48.99} \\
			{} & 0.298 & {42.2} & 5  &    & 0.00 & 8.70 & {8.70} & 4  & 140.79 & {352.63} & 5  &    & 0.00 & 0.20 & 0.20 & {44.94} \\
			\midrule
			{\multirow{2}[0]{*}{200}} & 0.107 & {56.8} & 5  &    & 0.00 & 7.93 & {7.93} & 5   & 420.52 & {420.52} & 5  &    & 0.00 & 0.21 & 0.21 & {54.51} \\
			{} & 0.304 & {57.2} & 5  &    & 0.00 & 36.53 & {36.53} & 5   & 293.05 & {293.05} & 5  &    & 0.00 & 0.20 & 0.20 & {54.27} \\
			\midrule
			{\multirow{2}[0]{*}{250}} & 0.112 & {71.4} & 5  &    & 0.00 & 30.25 & {30.25} & 5   & 1023.62 & {1023.62} & 5  &    & 0.00 & 0.19 & 0.19 & {49.82} \\
			{} & 0.304 & {71.2} & 5  &    & 0.00 & 85.98 & {85.98} & 5   & 712.18 & {712.18} & 5  &    & 0.00 & 0.42 & 0.42 & {51.83} \\
			\midrule
			{\multirow{2}[0]{*}{300}} & 0.110 & {86.6} & 5  &    & 0.00 & 74.87 & {74.87} & 0  & & {1200.00} & 5  &    & 0.00 & 0.23 & 0.23 & {57.38} \\
			{} & 0.302 & {88.6} & 5  &    & 0.00 & 340.37 & {340.37} & 4   & 726.59 & {821.27} & 5  &    & 0.00 & 1.04 & 1.04 & {62.28} \\
			\midrule
			{\multirow{2}[0]{*}{350}} & 0.117 & {102.6} & 5  &    & 0.00 & 204.47 & {204.47} & 0  &    & {1200.00} & 5  &    & 0.00 & 0.23 & 0.23 & {63.60} \\
			{} & 0.301 & {100.0} & 4  & 94.52 & 18.90 & 689.68 & {791.74} & 0  &    & {1200.00} & 5  &    & 0.00 & 0.51 & 0.51 & {69.82} \\
			\midrule
			{\multirow{2}[0]{*}{400}} & 0.111 & {114.8} & 5  &    & 0.00 & 265.87 & {265.87} & 0   &    & {1200.00} & 5  &    & 0.00 & 0.34 & 0.34 & {70.10} \\
			{} & 0.315 & {114.0} & 0  & 51.02 & 51.02 &    & {1200.00} & 0  &    & {1200.00} & 5  &    & 0.00 & 0.65 & 0.65 & {68.71} \\
			\midrule
			{\multirow{2}[0]{*}{450}} & 0.117 & {130.2} & 3  & 98.42 & 39.37 & 523.80 & {794.28} & 0   &    & {1200.00} & 5  &    & 0.00 & 0.39 & 0.39 & {69.22} \\
			{} & 0.309 & {130.4} & 0  & 99.17 & 99.17 &    & {1200.00} & 0   &    & {1200.00} & 5  &    & 0.00 & 2.02 & 2.02 & {76.74} \\
			\midrule
			{\multirow{2}[1]{*}{500}} & 0.117 & {142.8} & 1  & 98.59 & 78.87 & 1180.90 & {1196.18} & 0   &    & {1200.00} & 5  &    & 0.00 & 0.59 & 0.59 & {71.36} \\
			{} & 0.296 & {143.0} & 0  & 100.00 & 100.00 &    & {1200.00} & 0   &    & {1200.00} & 5  &    & 0.00 & 1.05 & 1.05 & {71.71} \\
			\midrule
			&    &    & {\bf 78} & {\bf \phantom{0}90.29} & {\bf \phantom{0}19.37} & {\bf \phantom{0}203.34} & {\bf \phantom{0}372.23} & {\bf 53} &  {\bf \phantom{0}320.05} & {\bf \phantom{0}731.35} & {\bf 100} & {\bf -} & {\bf \phantom{}0.00} & {\bf \phantom{}0.48} & {\bf \phantom{}0.48} & {\bf 57.91} \\			
			\bottomrule
		\end{tabular}%
	}
\end{sidewaystable}

\begin{sidewaystable}[htbp]
	\centering
	\caption{Experimental results for high-density perfect graph instances with small clusters \label{tab:perfectsmall_high}}
	\resizebox{\textwidth}{!}{
		\begin{tabular}{c c S[table-format=3.1] 
		c S[table-format=3.2] S[table-format=3.2] S[table-format=3.2] S[table-format=4.2]
		c S[table-format=4.2] S[table-format=4.2] 
		c S[table-format=2.2] S[table-format=2.2] S[table-format=3.2] S[table-format=4.2] S[table-format=2.2] }
			\toprule
			&    &    & \multicolumn{5}{c}{\textbf{IP formulation}} & \multicolumn{3}{c}{\textbf{B\&P}} & \multicolumn{6}{c}{\textbf{Cutting Plane}}\\
			\cmidrule(lr){4-8}\cmidrule(lr){9-11}\cmidrule(lr){12-17}
			\parbox[t]{1.3cm}{\centering $\boldsymbol{n}$}  &  \parbox[t]{1.3cm}{\centering \textbf{Avg\\density}} & \parbox[t]{1.3cm}{\centering \textbf{Avg \\ \# \\ clust}} & \parbox[t]{1.3cm}{\centering \textbf{\#\\opt}} & \parbox[t]{1.3cm}{\centering \textbf{Avg \\ \% gap \\in \\ nonopt}} & \parbox[t]{1.3cm}{\centering \textbf{Avg \\ \% gap \\overall}} & \parbox[t]{1.5cm}{\centering \textbf{Avg \\ time \\in opt}} & \parbox[t]{1.5cm}{\centering \textbf{Avg \\ time \\overall}} &
			\parbox[t]{1.3cm}{\centering \textbf{\#\\opt}} & \parbox[t]{1.5cm}{\centering \textbf{Avg \\ time \\in opt}} & \parbox[t]{1.5cm}{\centering \textbf{Avg \\ time \\overall}} &
			\parbox[t]{1.3cm}{\centering \textbf{\#\\opt}} & \parbox[t]{1.3cm}{\centering \textbf{Avg \\ \% gap \\in \\ nonopt}} & \parbox[t]{1.3cm}{\centering \textbf{Avg \\ \% gap \\overall}} & \parbox[t]{1.5cm}{\centering \textbf{Avg \\ time \\in opt}} & \parbox[t]{1.5cm}{\centering \textbf{Avg \\ time \\overall}} & \parbox[t]{1.3cm}{\centering \textbf{Avg \\ \% time \\in \\subpr}}   \\
			\midrule
			
			{\multirow{2}[2]{*}{50}} & 0.495 & 14.2 & 5  &    & 0.00 & 0.62 & 0.62 & 5   & 2.33 & 2.33 & 5  &    & 0.00 & 0.14 & 0.14 & 35.20 \\
			& 0.710 & 14.0 & 5  &    & 0.00 & 0.79 & 0.79 & 5   & 1.91 & 1.91 & 5  &    & 0.00 & 0.35 & 0.35 & 50.82 \\
			\midrule
			{\multirow{2}[2]{*}{100}} & 0.488 & 28.0 & 5  &    & 0.00 & 3.02 & 3.02 & 5   & 17.38 & 17.38 & 5  &    & 0.00 & 0.34 & 0.34 & 55.38 \\
			& 0.705 & 28.8 & 5  &    & 0.00 & 5.65 & 5.65 & 5   & 14.91 & 14.91 & 5  &    & 0.00 & 1.41 & 1.41 & 31.38 \\
			\midrule
			{\multirow{2}[2]{*}{150}} & 0.498 & 44.2 & 5  &    & 0.00 & 23.10 & 23.10 & 5   & 68.55 & 68.55 & 5  &    & 0.00 & 0.50 & 0.50 & 53.36 \\
			{} & 0.693 & 43.4 & 5  &    & 0.00 & 38.99 & 38.99 & 5   & 61.09 & 61.09 & 5  &    & 0.00 & 6.12 & 6.12 & 19.41 \\
			\midrule
			{\multirow{2}[2]{*}{200}} & 0.496 & 57.2 & 5  &    & 0.00 & 84.10 & 84.10 & 5   & 152.09 & 152.09 & 5  &    & 0.00 & 0.67 & 0.67 & 55.49 \\
			& 0.703 & 57.6 & 4  & 9.09 & 1.82 & 637.24 & 749.79 & 5   & 102.61 & 102.61 & 4  & 9.09 & 1.82 & 232.05 & 425.64 & 3.11 \\
			\midrule
			{\multirow{2}[2]{*}{250}} & 0.497 & {72.2} & 5  &    & 0.00 & 480.98 & 480.98 & 5  & 253.53 & 253.53 & 5  &    & 0.00 & 1.89 & 1.89 & 48.41 \\
			& 0.693 & 70.2 & 2  & 23.89 & 14.33 & 830.33 & 1052.13 & 5   & 215.79 & 215.79 & 2  & 13.33 & 8.00 & 254.25 & 821.70 & 1.94 \\
			\midrule
			{\multirow{2}[2]{*}{300}} & 0.506 & 83.4 & 2  & 24.68 & 14.81 & 500.43 & 920.17 & 5   & 507.82 & 507.82 & 5  &    & 0.00 & 4.72 & 4.72 & 26.70 \\
			& 0.691 & 85.6 & 0  & 83.96 & 83.96 &    & 1200.00 & 5  & 358.41 & 358.41 & 2  & 18.51 & 11.10 & 319.48 & 847.79 & 3.03 \\
			\midrule
			{\multirow{2}[2]{*}{350}} & 0.508 & 98.4 & 0  & 94.75 & 94.75 &    & 1200.00 & 5   & 657.70 & 657.70 & 4  & 8.33 & 1.67 & 15.55 & 252.44 & 20.48 \\
			& 0.698 & 99.6 & 0  & 89.50 & 89.50 &    & 1200.00 & 5  & 559.38 & 559.38 & 0  & 18.72 & 18.72 &    & 1200.00 & 3.26 \\
			\midrule
			{\multirow{2}[2]{*}{400}} & 0.502 & 114.2 & 0  & 86.91 & 86.91 &    & 1200.00 & 4   & 1092.62 & 1114.09 & 4  & 8.33 & 1.67 & 25.32 & 260.25 & 31.24 \\
			& 0.692 & 112.8 & 0  & 97.81 & 97.81 &    & 1200.00 & 5   & 742.60 & 742.60 & 0  & 24.55 & 24.55 &    & 1200.00 & 4.76 \\
			\midrule
			{\multirow{2}[2]{*}{450}} & 0.507 & 128.4 & 0  & 98.72 & 98.72 &    & 1200.00 & 0   &    & 1200.00 & 4  & 10.00 & 2.00 & 64.26 & 291.41 & 33.84 \\
			& 0.696 & 125.8 & 0  & 98.32 & 98.32 &    & 1200.00 & 2   & 975.35 & 1110.14 & 0  & 33.30 & 33.30 &    & 1200.00 & 6.73 \\
			\midrule
			{\multirow{2}[2]{*}{500}} & 0.507 & 144.2 & 0  & 99.59 & 99.59 &    & 1200.00 & 2   & 1193.47 & 1197.39 & 1  & 15.22 & 12.18 & 74.45 & 974.89 & 9.07 \\
			& 0.695 & 141.0 & 0  & 100.00 & 100.00 &    & 1200.00 & 5   & 1059.82 & 1059.82 & 0  & 30.15 & 30.15 &    & 1200.00 & 7.89 \\
			\midrule
			&    &    & {\bf 48} & {\bf \phantom{0}75.60} & {\bf \phantom{0}44.03} & {\bf 236.84} & {\bf \phantom{0}707.97} & {\bf 88} & {\bf \phantom{0}423.02} & {\bf \phantom{0}469.88} & {\bf 66} & {\bf 17.23} & {\bf \phantom{0}7.26} & {\bf \phantom{0}62.59} & {\bf \phantom{0}434.51} & {\bf 25.07} \\			
			\bottomrule
		\end{tabular}%
	}
\end{sidewaystable}%

Table \ref{tab:perfectsmall_high} presents the results of our experiments conducted on instances having edge density 0.5 and 0.7 (which are referred to as {\it high density}) with small clusters. The structure of this table is the same as Table \ref{tab:perfectsmall_low}.
From the results listed in these two tables, we observe that our approach yields superior results to both of the others in terms of solution time, and to IP formulation in terms of optimality gap as well, consistent with the aggregate results presented in Table \ref{tab:perfectsummary}. 
For low-density instances, the B\&P method fails to optimally solve any of the instances with 350 or more vertices, while the cutting plane method does not seem to be affected from increasing $n$ values and solves all optimally in under one second on the average.
While increased edge densities negatively affect the performances of the IP formulation and the cutting plane method, it improves that of the B\&P method so that higher $n$ values impact its performance much less than the low-density case. 
Even though there is no cut-off value for $n$ above which the cutting plane method fails to solve optimally, high densities nevertheless make its performance more sensitive to increasing $n$ values.

Next, we present a detailed summary of the experimental results obtained using the same collection of perfect graphs, but with different sets of clusters. In this case, the aim is to investigate the effect of increased cluster sizes on performances. We use the same collection of perfect graphs as before, but with different sets of clusters having 4--7 and 6--9 vertices in them.
Figures \ref{fig:comp_results_cl_4_7} and \ref{fig:comp_results_cl_6_9} illustrate a comparison between the three methods in terms of percentage of optimally solved instances and solution times. 
When the clusters have 6--9 vertices, the cutting plane method outperforms in all three respects, and when cluster sizes vary between 4 and 7, the outperformance of the cutting plane persists except for few cases.
Improved solvability of the cutting plane method also reveals as just slight variability in the average solution times with respect to $n$ values. 

\begin{figure}[h]
	\centering
	\begin{subfigure}{.33\textwidth}
		\centering
		\includegraphics[width=.99\linewidth]{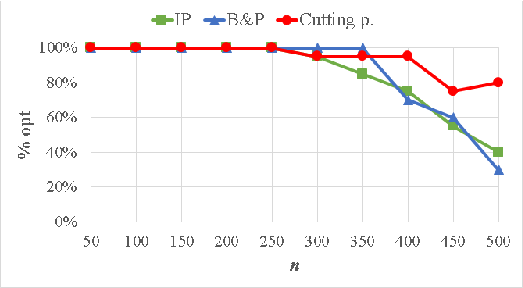}
		\caption{\footnotesize Percentage of optimally solved}
		\label{subfig_a:perc_opt_cl_4_7}
	\end{subfigure}%
	\begin{subfigure}{.33\textwidth}
		\centering
		\includegraphics[width=.99\linewidth]{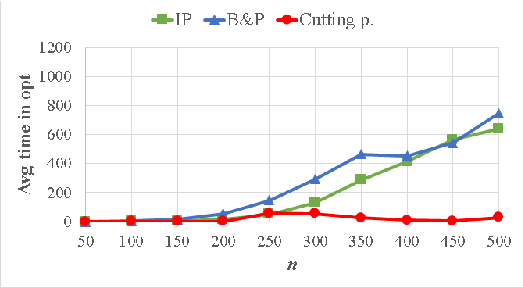}
		\caption{\footnotesize Average time in optimally solved}
		\label{subfig_b:time_opt_cl_4_7}
	\end{subfigure}
	\begin{subfigure}{.33\textwidth}
		\centering
		\includegraphics[width=.99\linewidth]{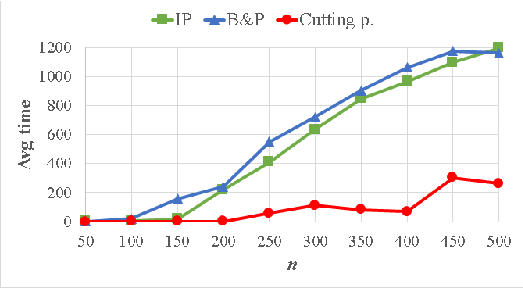}
		\caption{\footnotesize Average time over all instances}
		\label{subfig_c:time_all_cl_4_7}
	\end{subfigure}
    \caption{Results on perfect graph instances with 4-7 vertices in each cluster}
	\label{fig:comp_results_cl_4_7}
\end{figure}

\begin{figure}[h]
	\centering
	\begin{subfigure}{.33\textwidth}
		\centering
		\includegraphics[width=.99\linewidth]{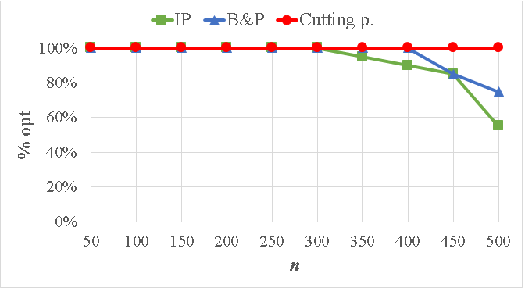}
		\caption{\footnotesize Percentage of optimally solved}
		\label{subfig_a:perc_opt_cl_6_9}
	\end{subfigure}%
	\begin{subfigure}{.33\textwidth}
		\centering
		\includegraphics[width=.99\linewidth]{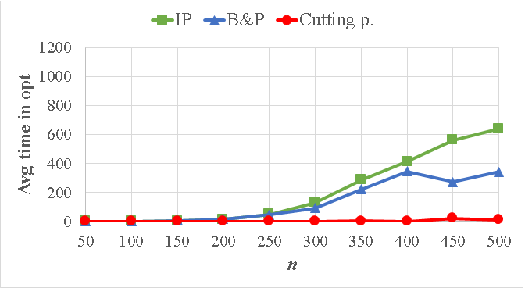}
		\caption{\footnotesize Average time in optimally solved}
		\label{subfig_b:time_opt_cl_6_9}
	\end{subfigure}
	\begin{subfigure}{.33\textwidth}
		\centering
		\includegraphics[width=.99\linewidth]{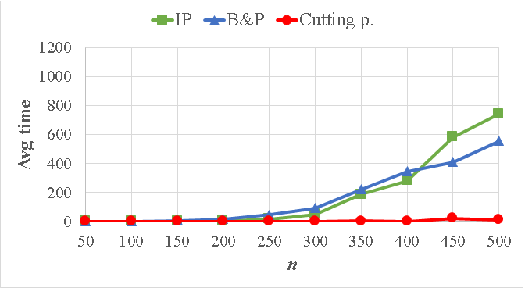}
		\caption{\footnotesize Average time over all instances}
		\label{subfig_c:time_all_cl_6_9}
	\end{subfigure}
    \caption{Results on perfect graph instances with 6-9 vertices in each cluster}
	\label{fig:comp_results_cl_6_9}
\end{figure}

\begin{sidewaystable}[htbp]
	\centering
	\caption{Experimental results for low-density perfect graph instances with medium-sized clusters \label{tab:perfectmedium_low}}
	\resizebox{\textwidth}{!}{
		\begin{tabular}{c c S[table-format=2.1] 
    		c S[table-format=3.2] S[table-format=2.2] S[table-format=3.2] S[table-format=4.2]
    		c S[table-format=3.2] S[table-format=4.2] 
    		c S[table-format=2.2] S[table-format=1.2] S[table-format=1.2] S[table-format=1.2] S[table-format=2.2] }
    			\toprule
    			&    &    & \multicolumn{5}{c}{\textbf{IP formulation}} & \multicolumn{3}{c}{\textbf{B\&P}} & \multicolumn{6}{c}{\textbf{Cutting Plane}}\\
    			\cmidrule(lr){4-8}\cmidrule(lr){9-11}\cmidrule(lr){12-17}
    			\parbox[t]{1.3cm}{\centering $\boldsymbol{n}$}  &  \parbox[t]{1.3cm}{\centering \textbf{Avg\\density}} & \parbox[t]{1.3cm}{\centering \textbf{Avg \\ \# \\ clust}} & \parbox[t]{1.3cm}{\centering \textbf{\#\\opt}} & \parbox[t]{1.3cm}{\centering \textbf{Avg \\ \% gap \\in \\ nonopt}} & \parbox[t]{1.3cm}{\centering \textbf{Avg \\ \% gap \\overall}} & \parbox[t]{1.5cm}{\centering \textbf{Avg \\ time \\in opt}} & \parbox[t]{1.5cm}{\centering \textbf{Avg \\ time \\overall}} &
    			\parbox[t]{1.3cm}{\centering \textbf{\#\\opt}} & \parbox[t]{1.5cm}{\centering \textbf{Avg \\ time \\in opt}} & \parbox[t]{1.5cm}{\centering \textbf{Avg \\ time \\overall}} &
    			\parbox[t]{1.3cm}{\centering \textbf{\#\\opt}} & \parbox[t]{1.3cm}{\centering \textbf{Avg \\ \% gap \\in \\ nonopt}} & \parbox[t]{1.3cm}{\centering \textbf{Avg \\ \% gap \\overall}} & \parbox[t]{1.5cm}{\centering \textbf{Avg \\ time \\in opt}} & \parbox[t]{1.5cm}{\centering \textbf{Avg \\ time \\overall}} & \parbox[t]{1.3cm}{\centering \textbf{Avg \\ \% time \\in \\subpr}}   \\
    			\midrule
				
				{\multirow{2}[1]{*}{50}} & 0.110 & 9.0 & 5  &    & 0.00 & 0.08 & 0.08 & 5   & 0.01 & 0.01 & 5  &    & 0.00 & 0.10 & 0.10 & 25.34 \\
				& 0.293 & 9.0 & 5  &    & 0.00 & 0.10 & 0.10 & 5   & 0.02 & 0.02 & 5  &    & 0.00 & 0.14 & 0.14 & 35.13 \\
				\midrule
				{\multirow{2}[0]{*}{100}} & 0.096 & 17.8 & 5  &    & 0.00 & 0.20 & 0.20 & 5   & 0.02 & 0.02 & 5  &    & 0.00 & 0.24 & 0.24 & 42.40 \\
				 & 0.300 & 17.8 & 5  &    & 0.00 & 1.04 & 1.04 & 5   & 0.06 & 0.06 & 5  &    & 0.00 & 0.26 & 0.26 & 55.77 \\
				\midrule
				{\multirow{2}[0]{*}{150}} & 0.098 & 27.6 & 5  &    & 0.00 & 0.76 & 0.76 & 5   & 0.08 & 0.08 & 5  &    & 0.00 & 0.38 & 0.38 & 57.27 \\
				& 0.298 & 26.6 & 5  &    & 0.00 & 1.59 & 1.59 & 5   & 0.09 & 0.09 & 5  &    & 0.00 & 0.41 & 0.41 & 58.39 \\
				\midrule
				{\multirow{2}[0]{*}{200}} & 0.107 & 36.8 & 5  &    & 0.00 & 1.71 & 1.71 & 5   & 0.86 & 0.86 & 5  &    & 0.00 & 0.46 & 0.46 & 69.35 \\
				& 0.304 & 36.0 & 5  &    & 0.00 & 7.41 & 7.41 & 5   & 50.36 & 50.36 & 5  &    & 0.00 & 0.70 & 0.70 & 75.50 \\
				\midrule
				{\multirow{2}[0]{*}{250}} & 0.112 & 45.0 & 5  &    & 0.00 & 4.60 & 4.60 & 5   & 1.29 & 1.29 & 5  &    & 0.00 & 0.43 & 0.43 & 62.99 \\
				& 0.304 & 45.8 & 5  &    & 0.00 & 16.56 & 16.56 & 5   & 243.26 & 243.26 & 5  &    & 0.00 & 0.78 & 0.78 & 77.36 \\
				\midrule
				{\multirow{2}[0]{*}{300}} & 0.110 & 54.4 & 5  &    & 0.00 & 10.13 & 10.13 & 5   & 14.05 & 14.05 & 5  &    & 0.00 & 0.46 & 0.46 & 73.36 \\
				& 0.302 & 54.2 & 5  &    & 0.00 & 94.55 & 94.55 & 5   & 475.23 & 475.23 & 5  &    & 0.00 & 2.17 & 2.17 & 72.51 \\
				\midrule
				{\multirow{2}[0]{*}{350}} & 0.117 & 63.6 & 5  &    & 0.00 & 22.88 & 22.88 & 5   & 163.62 & 163.62 & 5  &    & 0.00 & 1.21 & 1.21 & 84.68 \\
				& 0.301 & 62.8 & 5  &    & 0.00 & 223.55 & 223.55 & 5   & 671.28 & 671.28 & 5  &    & 0.00 & 1.50 & 1.50 & 82.32 \\
				\midrule
				{\multirow{2}[0]{*}{400}} & 0.111 & 71.4 & 5  &    & 0.00 & 33.68 & 33.68 & 5   & 221.57 & 221.57 & 5  &    & 0.00 & 1.06 & 1.06 & 81.53 \\
				& 0.315 & 73.0 & 5  &    & 0.00 & 494.34 & 494.34 & 0  &    & 1200.00 & 5  &    & 0.00 & 1.92 & 1.92 & 76.12 \\
				\midrule
				{\multirow{2}[0]{*}{450}} & 0.117 & 82.2 & 5  &    & 0.00 & 261.58 & 261.58 & 4   & 190.43 & 392.34 & 5  &    & 0.00 & 1.06 & 1.06 & 83.63 \\
				& 0.309 & 81.6 & 4  & 100.00 & 20.00 & 733.12 & 826.50 & 0  &    & 1200.00 & 5  &    & 0.00 & 1.87 & 1.87 & 77.57 \\
				\midrule
				{\multirow{2}[1]{*}{500}} & 0.117 & 90.6 & 5  &    & 0.00 & 449.23 & 449.23 & 1   & 4.32 & 960.86 & 5  &    & 0.00 & 1.90 & 1.90 & 87.44 \\
				& 0.296 & 91.4 & 3  & 87.50 & 35.00 & 961.09 & 1056.65 & 0   &    & 1200.00 & 5  &    & 0.00 & 2.17 & 2.17 & 84.64 \\
				\midrule
				&    &    & {\bf 97} & {\bf \phantom{0}93.75} & {\bf \phantom{0}2.75} & {\bf 165.91} & {\bf \phantom{0}175.36} & {\bf 80}  & {\bf 119.80} & {\bf \phantom{0}339.75} & {\bf 100} & {\bf -} & {\bf 0.00} & {\bf 0.96} & {\bf 0.96} & {\bf 68.16} \\				
				\bottomrule				
			\end{tabular}%
		}
\end{sidewaystable}%

\begin{sidewaystable}[htbp]
	\centering
	\caption{Experimental results for high-density perfect graph instances with medium-sized clusters	\label{tab:perfectmedium_high}}
	\resizebox{\textwidth}{!}{
		\begin{tabular}{ccccccSScSScccSSS}
    			\toprule
    			&    &    & \multicolumn{5}{c}{\textbf{IP formulation}} & \multicolumn{3}{c}{\textbf{B\&P}} & \multicolumn{6}{c}{\textbf{Cutting Plane}}\\
    			\cmidrule(lr){4-8}\cmidrule(lr){9-11}\cmidrule(lr){12-17}
    			\parbox[t]{1.3cm}{\centering $\boldsymbol{n}$}  &  \parbox[t]{1.3cm}{\centering \textbf{Avg\\density}} & \parbox[t]{1.3cm}{\centering \textbf{Avg \\ \# \\ clust}} & \parbox[t]{1.3cm}{\centering \textbf{\#\\opt}} & \parbox[t]{1.3cm}{\centering \textbf{Avg \\ \% gap \\in \\ nonopt}} & \parbox[t]{1.3cm}{\centering \textbf{Avg \\ \% gap \\overall}} & \parbox[t]{1.3cm}{\centering \textbf{Avg \\ time \\in opt}} & \parbox[t]{1.3cm}{\centering \textbf{Avg \\ time \\overall}} &
    			\parbox[t]{1.3cm}{\centering \textbf{\#\\opt}} & \parbox[t]{1.3cm}{\centering \textbf{Avg \\ time \\in opt}} & \parbox[t]{1.3cm}{\centering \textbf{Avg \\ time \\overall}} &
    			\parbox[t]{1.3cm}{\centering \textbf{\#\\opt}} & \parbox[t]{1.3cm}{\centering \textbf{Avg \\ \% gap \\in \\ nonopt}} & \parbox[t]{1.3cm}{\centering \textbf{Avg \\ \% gap \\overall}} & \parbox[t]{1.3cm}{\centering \textbf{Avg \\ time \\in opt}} & \parbox[t]{1.3cm}{\centering \textbf{Avg \\ time \\overall}} & \parbox[t]{1.3cm}{\centering \textbf{Avg \\ \% time \\in \\subpr}}   \\
    			
				\midrule
				{\multirow{2}[2]{*}{50}} & 0.495 & {9.0} & 5  &    & 0.00 & 0.15 & {0.15} & 5  & 0.04 & {0.04} & 5  &    & 0.00 & 0.22 & 0.22 & {45.21} \\
				{} & 0.710 & {8.8} & 5  &    & 0.00 & 0.32 & {0.32} & 5   & 1.04 & {1.04} & 5  &    & 0.00 & 0.29 & 0.29 & {55.70} \\
				\midrule
				{\multirow{2}[2]{*}{100}} & 0.488 & {17.6} & 5  &    & 0.00 & 3.20 & {3.20} & 5  & 3.44 & {3.44} & 5  &    & 0.00 & 0.47 & 0.47 & {66.76} \\
				{} & 0.705 & {18.0} & 5  &    & 0.00 & 1.62 & {1.62} & 5   & 9.91 & {9.91} & 5  &    & 0.00 & 1.64 & 1.64 & {60.08} \\
				\midrule
				{\multirow{2}[2]{*}{150}} & 0.498 & {26.6} & 5  &    & 0.00 & 7.38 & {7.38} & 5   & 20.24 & {20.24} & 5  &    & 0.00 & 0.57 & 0.57 & {73.52} \\
				{} & 0.693 & {27.8} & 5  &    & 0.00 & 8.94 & {8.94} & 5   & 37.06 & {37.06} & 5  &    & 0.00 & 2.99 & 2.99 & {52.97} \\
				\midrule
				{\multirow{2}[2]{*}{200}} & 0.496 & {35.6} & 5  &    & 0.00 & 19.81 & {19.81} & 5  & 98.02 & {98.02} & 5  &    & 0.00 & 1.08 & 1.08 & {70.71} \\
				{} & 0.703 & {36.2} & 5  &    & 0.00 & 33.79 & {33.79} & 5  & 76.94 & {76.94} & 5  &    & 0.00 & 4.43 & 4.43 & {44.63} \\
				\midrule
				{\multirow{2}[2]{*}{250}} & 0.497 & {45.2} & 5  &    & 0.00 & 58.95 & {58.95} & 5   & 209.77 & {209.77} & 5  &    & 0.00 & 0.65 & 0.65 & {70.86} \\
				{} & 0.693 & {46.4} & 5  &    & 0.00 & 121.05 & {121.05} & 5  & 132.32 & {132.32} & 5  &    & 0.00 & 235.66 & 235.66 & {7.86} \\
				\midrule
				{\multirow{2}[2]{*}{300}} & 0.506 & {55.0} & 5  &    & 0.00 & 195.94 & {195.94} & 5   & 450.78 & {450.78} & 5  &    & 0.00 & 0.98 & 0.98 & {74.14} \\
				{} & 0.691 & {54.4} & 4  & 20.00 & 4.00 & 250.94 & {440.76} & 5   & 226.42 & {226.42} & 4  & 20.00 & 4.00 & 257.77 & 446.21 & {23.36} \\
				\midrule
				{\multirow{2}[2]{*}{350}} & 0.508 & {63.8} & 4  & 100.00 & 20.00 & 407.27 & {565.81} & 5   & 734.28 & {734.28} & 5  &    & 0.00 & 1.71 & 1.71 & {67.28} \\
				{} & 0.698 & {63.2} & 3  & 60.00 & 24.00 & 670.25 & {882.15} & 5   & 287.89 & {287.89} & 4  & 16.67 & 3.33 & 111.61 & 329.29 & {9.15} \\
				\midrule
				{\multirow{2}[2]{*}{400}} & 0.502 & {73.0} & 5  &    & 0.00 & 713.59 & {713.59} & 4   & 640.53 & {752.42} & 5  &    & 0.00 & 3.72 & 3.72 & {54.37} \\
				{} & 0.692 & {73.2} & 0  & 82.86 & 82.86 &    & {1200.00} & 5  & 539.49 & {539.49} & 4  & 37.50 & 7.50 & 33.92 & 267.14 & {20.36} \\
				\midrule
				{\multirow{2}[2]{*}{450}} & 0.507 & {82.8} & 2  & 96.67 & 58.00 & 981.34 & {1112.53} & 3   & 980.21 & {1068.12} & 5  &    & 0.00 & 11.30 & 11.30 & {62.12} \\
				{} & 0.696 & {81.2} & 0  & 95.33 & 95.33 &    & {1200.00} & 5   & 564.83 & {564.83} & 0  & 31.07 & 31.07 &    & 1200.00 & {4.26} \\
				\midrule
				{\multirow{2}[2]{*}{500}} & 0.507 & {91.2} & 0  & 99.20 & 99.20 &    & {1200.00} & 0   &    & {1200.00} & 5  &    & 0.00 & 6.62 & 6.62 & {54.75} \\
				{} & 0.695 & {92.6} & 0  & 100.00 & 100.00 &    & {1200.00} & 5   & 895.48 & {895.48} & 1  & 31.80 & 25.44 & 406.27 & 1041.25 & {3.80} \\
				\midrule
				&    &    & \textbf{73} & \textbf{81.76} & \textbf{24.17} & \textbf{217.16} & \textbf{448.30} & \textbf{92} & \textbf{310.98} & \textbf{365.42} & \textbf{88} & \textbf{27.41} & \textbf{3.57} & \textbf{56.94} & \textbf{177.81} & \textbf{46.09} \\
				\bottomrule				
			\end{tabular}%
		}
	{}
\end{sidewaystable}%

\begin{sidewaystable}[htbp]
	\centering
	\caption{Experimental results for low-density perfect graph instances with large clusters
	\label{tab:perfectlarge_low}}
	\resizebox{\textwidth}{!}{
		\begin{tabular}{c c S[table-format=2.1] 
    		c S[table-format=3.2] S[table-format=1.2] S[table-format=3.2] S[table-format=3.2]
    		c S[table-format=2.2] S[table-format=2.2] 
    		c S[table-format=2.2] S[table-format=1.2] S[table-format=1.2] S[table-format=1.2] S[table-format=2.2] }
    			\toprule
    			&    &    & \multicolumn{5}{c}{\textbf{IP formulation}} & \multicolumn{3}{c}{\textbf{B\&P}} & \multicolumn{6}{c}{\textbf{Cutting Plane}}\\
    			\cmidrule(lr){4-8}\cmidrule(lr){9-11}\cmidrule(lr){12-17}
    			\parbox[t]{1.3cm}{\centering $\boldsymbol{n}$}  &  \parbox[t]{1.3cm}{\centering \textbf{Avg\\density}} & \parbox[t]{1.3cm}{\centering \textbf{Avg \\ \# \\ clust}} & \parbox[t]{1.3cm}{\centering \textbf{\#\\opt}} & \parbox[t]{1.3cm}{\centering \textbf{Avg \\ \% gap \\in \\ nonopt}} & \parbox[t]{1.3cm}{\centering \textbf{Avg \\ \% gap \\overall}} & \parbox[t]{1.5cm}{\centering \textbf{Avg \\ time \\in opt}} & \parbox[t]{1.5cm}{\centering \textbf{Avg \\ time \\overall}} &
    			\parbox[t]{1.3cm}{\centering \textbf{\#\\opt}} & \parbox[t]{1.5cm}{\centering \textbf{Avg \\ time \\in opt}} & \parbox[t]{1.5cm}{\centering \textbf{Avg \\ time \\overall}} &
    			\parbox[t]{1.3cm}{\centering \textbf{\#\\opt}} & \parbox[t]{1.3cm}{\centering \textbf{Avg \\ \% gap \\in \\ nonopt}} & \parbox[t]{1.3cm}{\centering \textbf{Avg \\ \% gap \\overall}} & \parbox[t]{1.5cm}{\centering \textbf{Avg \\ time \\in opt}} & \parbox[t]{1.5cm}{\centering \textbf{Avg \\ time \\overall}} & \parbox[t]{1.3cm}{\centering \textbf{Avg \\ \% time \\in \\subpr}}   \\
    			\midrule
				
				{\multirow{2}[1]{*}{50}} & 0.110 & 6.6 & 5  &    & 0.00 & 0.06 & 0.06 & 5   & 0.01 & 0.01 & 5  &    & 0.00 & 0.17 & 0.17 & 19.72 \\
				& 0.293 & 6.6 & 5  &    & 0.00 & 0.08 & 0.08 & 5   & 0.01 & 0.01 & 5  &    & 0.00 & 0.19 & 0.19 & 29.79 \\
				\midrule
				{\multirow{2}[0]{*}{100}} & 0.096 & 13.4 & 5  &    & 0.00 & 0.16 & 0.16 & 5   & 0.02 & 0.02 & 5  &    & 0.00 & 0.10 & 0.10 & 30.79 \\
				{} & 0.300 & 13.2 & 5  &    & 0.00 & 0.72 & 0.72 & 5   & 0.03 & 0.03 & 5  &    & 0.00 & 0.16 & 0.16 & 43.86 \\
				\midrule
				{\multirow{2}[0]{*}{150}} & 0.098 & 19.4 & 5  &    & 0.00 & 0.53 & 0.53 & 5  & 0.08 & 0.08 & 5  &    & 0.00 & 0.23 & 0.23 & 38.11 \\
				& 0.298 & 19.8 & 5  &    & 0.00 & 1.01 & 1.01 & 5  & 0.05 & 0.05 & 5  &    & 0.00 & 0.38 & 0.38 & 53.32 \\
				\midrule
				{\multirow{2}[0]{*}{200}} & 0.107 & 26.8 & 5  &    & 0.00 & 1.03 & 1.03 & 5   & 0.11 & 0.11 & 5  &    & 0.00 & 0.38 & 0.38 & 53.28 \\
				& 0.304 & 26.6 & 5  &    & 0.00 & 3.37 & 3.37 & 5   & 0.14 & 0.14 & 5  &    & 0.00 & 0.50 & 0.50 & 64.61 \\
				\midrule
				{\multirow{2}[0]{*}{250}} & 0.112 & 33.2 & 5  &    & 0.00 & 2.67 & 2.67 & 5   & 0.24 & 0.24 & 5  &    & 0.00 & 0.23 & 0.23 & 50.76 \\
				& 0.304 & 33.4 & 5  &    & 0.00 & 9.18 & 9.18 & 5   & 0.36 & 0.36 & 5  &    & 0.00 & 0.37 & 0.37 & 69.05 \\
				\midrule
				{\multirow{2}[0]{*}{300}} & 0.110 & 40.2 & 5  &    & 0.00 & 6.15 & 6.15 & 5   & 0.52 & 0.52 & 5  &    & 0.00 & 0.30 & 0.30 & 53.70 \\
				& 0.302 & 40.2 & 5  &    & 0.00 & 22.06 & 22.06 & 5   & 4.85 & 4.85 & 5  &    & 0.00 & 0.63 & 0.63 & 71.75 \\
				\midrule
				{\multirow{2}[0]{*}{350}} & 0.117 & 46.2 & 5  &    & 0.00 & 14.54 & 14.54 & 5   & 0.75 & 0.75 & 5  &    & 0.00 & 0.29 & 0.29 & 60.45 \\
				& 0.301 & 47.6 & 5  &    & 0.00 & 52.97 & 52.97 & 5   & 85.89 & 85.89 & 5  &    & 0.00 & 0.72 & 0.72 & 70.16 \\
				\midrule
				{\multirow{2}[0]{*}{400}} & 0.111 & 53.2 & 5  &    & 0.00 & 22.06 & 22.06 & 5   & 1.24 & 1.24 & 5  &    & 0.00 & 0.43 & 0.43 & 66.04 \\
				& 0.315 & 52.6 & 5  &    & 0.00 & 113.07 & 113.07 & 5   & 1.51 & 1.51 & 5  &    & 0.00 & 0.90 & 0.90 & 81.08 \\
				\midrule
				{\multirow{2}[0]{*}{450}} & 0.117 & 59.0 & 5  &    & 0.00 & 137.30 & 137.30 & 5 & 1.90 & 1.90 & 5  &    & 0.00 & 0.47 & 0.47 & 69.24 \\
				& 0.309 & 59.4 & 5  &    & 0.00 & 439.23 & 439.23 & 5   & 44.07 & 44.07 & 5  &    & 0.00 & 1.24 & 1.24 & 84.12 \\
				\midrule
				{\multirow{2}[1]{*}{500}} & 0.117 & 65.8 & 5  &    & 0.00 & 143.72 & 143.72 & 5   & 2.71 & 2.71 & 5  &    & 0.00 & 0.73 & 0.73 & 76.06 \\
				& 0.296 & 66.4 & 5  &    & 0.00 & 480.92 & 480.92 & 5  & 42.40 & 42.40 & 5  &    & 0.00 & 0.94 & 0.94 & 82.08 \\
				\midrule
				&    &    & {\bf 100} & {\bf -} & {\bf 0.00} & {\bf \phantom{0}72.54} & {\bf \phantom{0}72.54} & {\bf 100} & {\bf \phantom{0}9.34} & {\bf \phantom{0}9.34} & {\bf 100} & {\bf -} & {\bf 0.00} & {\bf 0.47} & {\bf 0.47} & {\bf 58.40} \\			
				\bottomrule			
			\end{tabular}%
		}
\end{sidewaystable}%

\begin{sidewaystable}[htbp]
	\centering
	\caption{Experimental results for high-density perfect graph instances with large clusters	\label{tab:perfectlarge_high}}
	\resizebox{\textwidth}{!}{
	\begin{tabular}{c c S[table-format=2.1] 
		c S[table-format=3.2] S[table-format=3.2] S[table-format=3.2] S[table-format=4.2]
		c S[table-format=4.2] S[table-format=4.2] 
		c S[table-format=2.2] S[table-format=1.2] S[table-format=2.2] S[table-format=2.2] S[table-format=2.2] }
			\toprule
			&    &    & \multicolumn{5}{c}{\textbf{IP formulation}} & \multicolumn{3}{c}{\textbf{B\&P}} & \multicolumn{6}{c}{\textbf{Cutting Plane}}\\
			\cmidrule(lr){4-8}\cmidrule(lr){9-11}\cmidrule(lr){12-17}
			\parbox[t]{1.3cm}{\centering $\boldsymbol{n}$}  &  \parbox[t]{1.3cm}{\centering \textbf{Avg\\density}} & \parbox[t]{1.3cm}{\centering \textbf{Avg \\ \# \\ clust}} & \parbox[t]{1.3cm}{\centering \textbf{\#\\opt}} & \parbox[t]{1.3cm}{\centering \textbf{Avg \\ \% gap \\in \\ nonopt}} & \parbox[t]{1.3cm}{\centering \textbf{Avg \\ \% gap \\overall}} & \parbox[t]{1.5cm}{\centering \textbf{Avg \\ time \\in opt}} & \parbox[t]{1.5cm}{\centering \textbf{Avg \\ time \\overall}} &
			\parbox[t]{1.3cm}{\centering \textbf{\#\\opt}} & \parbox[t]{1.5cm}{\centering \textbf{Avg \\ time \\in opt}} & \parbox[t]{1.5cm}{\centering \textbf{Avg \\ time \\overall}} &
			\parbox[t]{1.3cm}{\centering \textbf{\#\\opt}} & \parbox[t]{1.3cm}{\centering \textbf{Avg \\ \% gap \\in \\ nonopt}} & \parbox[t]{1.3cm}{\centering \textbf{Avg \\ \% gap \\overall}} & \parbox[t]{1.5cm}{\centering \textbf{Avg \\ time \\in opt}} & \parbox[t]{1.5cm}{\centering \textbf{Avg \\ time \\overall}} & \parbox[t]{1.3cm}{\centering \textbf{Avg \\ \% time \\in \\subpr}}   \\
			\midrule
			
			{\multirow{2}[2]{*}{50}} & 0.495 & 6.4 & 5  &    & 0.00 & 0.11 & 0.11 & 5   & 0.01 & 0.01 & 5  &    & 0.00 & 0.23 & 0.23 & 33.81 \\
			& 0.710 & 7.0 & 5  &    & 0.00 & 0.25 & 0.25 & 5   & 0.38 & 0.38 & 5  &    & 0.00 & 0.26 & 0.26 & 59.92 \\
			\midrule
			{\multirow{2}[2]{*}{100}} & 0.488 & 13.2 & 5  &    & 0.00 & 2.13 & 2.13 & 5  & 0.05 & 0.05 & 5  &    & 0.00 & 0.22 & 0.22 & 53.08 \\
			& 0.705 & 13.4 & 5  &    & 0.00 & 4.12 & 4.12 & 5   & 5.83 & 5.83 & 5  &    & 0.00 & 0.77 & 0.77 & 69.83 \\
			\midrule
			{\multirow{2}[2]{*}{150}} & 0.498 & 19.0 & 5  &    & 0.00 & 1.94 & 1.94 & 5   & 0.62 & 0.62 & 5  &    & 0.00 & 0.35 & 0.35 & 64.90 \\
			& 0.693 & 20.2 & 5  &    & 0.00 & 4.65 & 4.65 & 5   & 26.65 & 26.65 & 5  &    & 0.00 & 0.93 & 0.93 & 61.71 \\
			\midrule
			{\multirow{2}[2]{*}{200}} & 0.496 & 26.6 & 5  &    & 0.00 & 7.13 & 7.13 & 5   & 12.97 & 12.97 & 5  &    & 0.00 & 0.98 & 0.98 & 58.60 \\
			& 0.703 & 26.8 & 5  &    & 0.00 & 15.67 & 15.67 & 5   & 55.31 & 55.31 & 5  &    & 0.00 & 4.13 & 4.13 & 51.20 \\
			\midrule
			{\multirow{2}[2]{*}{250}} & 0.497 & 33.2 & 5  &    & 0.00 & 19.66 & 19.66 & 5   & 62.03 & 62.03 & 5  &    & 0.00 & 0.79 & 0.79 & 77.09 \\
			& 0.693 & 33.4 & 5  &    & 0.00 & 42.21 & 42.21 & 5   & 129.70 & 129.70 & 5  &    & 0.00 & 4.05 & 4.05 & 51.07 \\
			\midrule
			{\multirow{2}[2]{*}{300}} & 0.506 & 40.0 & 5  &    & 0.00 & 47.66 & {47.66} & 5   & 123.92 & 123.92 & 5  &    & 0.00 & 1.36 & 1.36 & 80.21 \\
			& 0.691 & 40.6 & 5  &    & 0.00 & 118.47 & 118.47 & 5  & 244.15 & 244.15 & 5  &    & 0.00 & 5.04 & 5.04 & 44.38 \\
			\midrule
			{\multirow{2}[2]{*}{350}} & 0.508 & 46.6 & 5  &    & 0.00 & 274.99 & 274.99 & 5   & 394.37 & 394.37 & 5  &    & 0.00 & 2.25 & 2.25 & 82.15 \\
			& 0.698 & 46.8 & 4  & 100.00 & 20.00 & 208.45 & 406.76 & 5  & 407.26 & 407.26 & 5  &    & 0.00 & 8.06 & 8.06 & 40.88 \\
			\midrule
			{\multirow{2}[2]{*}{400}} & 0.502 & 52.8 & 5  &    & 0.00 & 298.38 & 298.38 & 5   & 728.67 & 728.67 & 5  &    & 0.00 & 2.58 & 2.58 & 64.96 \\
			{} & 0.692 & 52.2 & 3  & 100.00 & 40.00 & 354.49 & 692.69 & 5  & 658.65 & 658.65 & 5  &    & 0.00 & 4.23 & 4.23 & 46.87 \\
			\midrule
			{\multirow{2}[2]{*}{450}} & 0.507 & 59.6 & 3  & 100.00 & 40.00 & 538.64 & 803.18 & 2   & 372.79 & 869.11 & 5  &    & 0.00 & 1.92 & 1.92 & 81.49 \\
			& 0.696 & 60.6 & 4  & 100.00 & 20.00 & 896.27 & 957.01 & 5   & 733.06 & 733.06 & 5  &    & 0.00 & 85.12 & 85.12 & 12.77 \\
			\midrule
			{\multirow{2}[2]{*}{500}} & 0.507 & 66.0 & 1  & 100.00 & 80.00 & 960.58 & 1152.12 & 2   & 903.36 & 1081.34 & 5  &    & 0.00 & 4.65 & 4.65 & 69.60 \\
			& 0.695 & 65.0 & 0  & 100.00 & 100.00 &    & 1200.00 & 3   & 1038.45 & 1103.07 & 5  &    & 0.00 & 40.48 & 40.48 & 25.01 \\
			\midrule 
			&    &    & {\bf 85} & {\bf 100.00} & {\bf \phantom{0}15.00} & {\bf 199.78} & {\bf \phantom{0}302.46} & {\bf 92}  & {\bf \phantom{0}294.91} & {\bf \phantom{0}331.86} & {\bf 100} & {\bf -} & {\bf 0.00} & {\bf \phantom{0}8.42} & {\bf \phantom{0}8.42} & {\bf 56.48} \\
			\bottomrule				
		\end{tabular}%
	}
\end{sidewaystable}%

A more detailed summary of experimental results for instances with cluster sizes 4--7 and 6--9 are provided in Tables \ref{tab:perfectmedium_low} to \ref{tab:perfectlarge_high}. When we compare the results in Tables \ref{tab:perfectsmall_low}--\ref{tab:perfectsmall_high} to those in Tables \ref{tab:perfectmedium_low}--\ref{tab:perfectlarge_high}, we notice that the performances improve with increased cluster sizes and hence decreased number of clusters. 
For a given $n$ value, the number of variables and constraints in the IP formulation decreases with the increase in the average size of the clusters, which leads to a shrinkage in problem size and hence to improved performance for IP.
We observe that, for high densities, the largest $n$ value for which IP can optimally solve at least one instance rises from 300 to 450 and 500 in instances with medium-sized and large clusters, respectively.
As for the B\&P method, there is no $n$ value above which it cannot deliver an optimal solution anymore, even for the low-density instances where it is weaker.
In case of medium-sized clusters and low densities, the cutting plane algorithm outperforms the other two by optimally solving all instances in under one second, whereas IP and B\&P do so for respectively 97\% and 80\% of these instances while yielding considerably longer solution times as well. 
When the clusters are large, the cutting plane method is able to deliver optimal solutions for all the instances within a few seconds most of the time, and it becomes the outperforming method even in high densities.

As an additional set of experiments, we compare the performance of our cutting plane algorithm for general perfect graphs with that designed for the subclasses of perfect graphs investigated in \citep{seker2019decomposition}. The subproblems in these three graph classes, which are permutation, generalized split, and chordal graphs, were solved via specialized combinatorial algorithms that are polynomial-time, whereas the maximum clique algorithm MCS by \cite{tomita2010simple} is not so, though it runs quite efficiently in practice. Using the same experimental environment as in \citep{seker2019decomposition}, we run our cutting plane algorithm for general perfect graphs on the test instances from the three subclasses of perfect graphs. The number of vertices of these instances range from 100 to 500 for permutation and generalized split graphs, and from 100 to 1000 for chordal graphs. The average edge densities are the same as here; namely, 0.1, 0.3, 0.5, and 0.7. The total number of instances tested are 1200 for chordal graphs, and 600 for the other two classes.

Table \ref{tab:permgrsummary} summarizes the results for permutation graphs. The structure of this table is the same as Table \ref{tab:perfectsummary}, except that the two sets of columns list the results for our algorithms for the general and special cases, respectively. Our first observation from this table is that the general algorithm surprisingly yields better results than the one tailored for permutation graphs. The improvement is particularly evident in high-density instances in terms of all three measures we list here. The average percentage gap value in high-density instances with small clusters drops from 23.32\% to 12.98\%, and the overall average of optimality gap improves by 5\%. 

\begin{table}[htbp]
	\caption{Summary of experimental results for permutation graph instances}	\label{tab:permgrsummary}
	\centering
		\resizebox{0.8\textwidth}{!}{
			\begin{tabular}{cc S[table-format=3.0] S[table-format=2.2] S[table-format=3.2] c S[table-format=3.0] S[table-format=2.2] S[table-format=3.2]}
				\toprule
				&    &  \multicolumn{3}{c}{\textbf{Cutting plane}} & & \multicolumn{3}{c}{\textbf{Decomp for perm gr \citep{seker2019decomposition}}}\\
				\cmidrule(lr){3-6}\cmidrule(lr){7-9}
				\parbox[t]{1.8cm}{\centering \textbf{Sizes of \\clusters}} &\parbox[t]{1.8cm}{\centering \textbf{Density}} & \parbox[t]{2cm}{\centering \textbf{\#\\opt}} & \parbox[t]{2cm}{\centering \textbf{Avg \\ \% gap}} & \parbox[t]{2.2cm}{\centering \textbf{Avg \\ time }} && \parbox[t]{2cm}{\centering \textbf{\#\\opt}} & \parbox[t]{2cm}{\centering \textbf{Avg \\ \% gap}} & \parbox[t]{2.2cm}{\centering \textbf{Avg \\ time }} \\
				\midrule
				
				\multirow{3}[2]{*}{small} & low & 99 & 0.25 & 14.41 &    & 98 & 0.40 & 47.16 \\
				& high & 52 & 12.98 & 641.35 &    & 37 & 23.32 & 780.47 \\[0.05cm]
				& all & 151 & 6.62 & 327.88 &    & 135 & 11.86 & 413.81 \\
				\midrule
				\multirow{3}[2]{*}{medium} & low & 100 & 0.00 & 1.18 &    & 100 & 0.00 & 3.81 \\
				& high & 53 & 19.44 & 649.30 &    & 42 & 29.78 & 753.63 \\[0.05cm]
				& all & 153 & 9.72 & 325.24 &    & 142 & 14.89 & 378.72 \\
				\midrule
				\multirow{3}[2]{*}{large} & low & 100 & 0.00 & 0.93 &    & 100 & 0.00 & 1.08 \\
				& high & 64 & 18.07 & 482.31 &    & 52 & 26.95 & 627.73 \\[0.05cm]
				& all & 164 & 9.03 & 241.62 &    & 152 & 13.47 & 314.41 \\
				\midrule
				\multicolumn{2}{c}{} & {\bf 468} & {\bf \phantom{0}8.46} & {\bf 298.25} &    & {\bf 429} & {\bf 13.41} & {\bf 368.98} \\
									
				\bottomrule				
			\end{tabular}%
	}
{}
\end{table}%

Table \ref{tab:gsgsummary} contains the summary of the results we obtained for generalized split graphs and has the same format as the previous one. As opposed to the case of permutation graphs, there is no monotonicity in the change of the performance between the cutting plane and the decomposition algorithm, not even within a given density or cluster size category. The overall performance deteriorates when we use our algorithm for general perfect graphs, but it is still comparable to that of the algorithm tailored for generalized split graphs.

\begin{table}[H]
	\caption{Summary of experimental results for generalized split graph instances}
	\label{tab:gsgsummary}
	\centering
	   	\resizebox{0.8\textwidth}{!}{
		\begin{tabular}{cc S[table-format=3.0] S[table-format=2.2] S[table-format=3.2] c S[table-format=3.0] S[table-format=2.2] S[table-format=3.2]}
				\toprule
				&    &  \multicolumn{3}{c}{\textbf{Cutting plane}} & & \multicolumn{3}{c}{\textbf{Decomp for GSG \citep{seker2019decomposition}}}\\
				\cmidrule(lr){3-6}\cmidrule(lr){7-9}
				\parbox[t]{1.8cm}{\centering \textbf{Sizes of \\clusters}} &\parbox[t]{1.8cm}{\centering \textbf{Density}} & \parbox[t]{2cm}{\centering \textbf{\#\\opt}} & \parbox[t]{2cm}{\centering \textbf{Avg \\ \% gap}} & \parbox[t]{2.2cm}{\centering \textbf{Avg \\ time }} && \parbox[t]{2cm}{\centering \textbf{\#\\opt}} & \parbox[t]{2cm}{\centering \textbf{Avg \\ \% gap}} & \parbox[t]{2.2cm}{\centering \textbf{Avg \\ time }} \\
				\midrule
				
				\multirow{3}[2]{*}{small} & low & 81 & 4.95 & 234.61 &    & 76 & 5.50 & 310.33 \\
				& high & 62 & 10.26 & 495.17 &    & 66 & 7.58 & 427.44 \\[0.05cm]
				& all & 143 & 7.61 & 364.89 &    & 142 & 6.54 & 368.88 \\
				\midrule
				\multirow{3}[2]{*}{medium} & low & 74 & 10.75 & 321.26 &    & 77 & 9.03 & 287.05 \\
				& high & 57 & 16.59 & 547.13 &    & 64 & 11.63 & 454.85 \\[0.05cm]
				& all & 131 & 13.67 & 434.19 &    & 141 & 10.33 & 370.95 \\
				\midrule
				\multirow{3}[2]{*}{large} & low & 70 & 13.67 & 393.82 &    & 72 & 12.55 & 364.49 \\
				& high & 67 & 15.33 & 443.80 &    & 66 & 12.66 & 431.26 \\[0.05cm]
				& all & 137 & 14.50 & 418.81 &    & 138 & 12.60 & 397.88 \\
				\midrule
				\multicolumn{2}{c}{} & {\bf 411} & {\bf 11.92} & {\bf 405.96} &    & {\bf 421} & {\bf \phantom{0}9.82} & {\bf 379.24} \\				
				\bottomrule				
			\end{tabular}%
	}
{}
\end{table}%

In the last part of this final group of experiments, we present the computational results for chordal graphs in Table \ref{tab:chordalsummary} in the same structure as the previous two. The algorithm we present in \citep{seker2019decomposition} yields the best results in the class of chordal graphs by solving all of the instances to optimality in approximately 0.16 seconds. In this case, the difference between the two methods is clear; the one custom-tailored for chordal graphs clearly outperforms in all respects. 
\begin{table}[H]
	\caption{Summary of experimental results for chordal graph instances}
	\label{tab:chordalsummary}
	\centering
		\resizebox{0.8\textwidth}{!}{
		\begin{tabular}{cc S[table-format=3.0] S[table-format=2.2] S[table-format=4.2] c S[table-format=3.0] S[table-format=1.2] S[table-format=1.2]}
				\toprule
				&    &  \multicolumn{3}{c}{\textbf{Cutting plane}} & & \multicolumn{3}{c}{\textbf{Decomp for chordal gr \citep{seker2019decomposition}}}\\
				\cmidrule(lr){3-6}\cmidrule(lr){7-9}
				\parbox[t]{1.8cm}{\centering \textbf{Sizes of \\clusters}} &\parbox[t]{1.8cm}{\centering \textbf{Density}} & \parbox[t]{2cm}{\centering \textbf{\#\\opt}} & \parbox[t]{2cm}{\centering \textbf{Avg \\ \% gap}} & \parbox[t]{2.3cm}{\centering \textbf{Avg \\ time }} && \parbox[t]{2cm}{\centering \textbf{\#\\opt}} & \parbox[t]{2cm}{\centering \textbf{Avg \\ \% gap}} & \parbox[t]{2.4cm}{\centering \textbf{Avg \\ time }} \\
				\midrule
				
				\multirow{3}[2]{*}{small} & low   & 53    & 25.56 & 897.58 &       & 200   & 0.00  & 0.20 \\
				& high  & 17    & 43.65 & 1113.71 &       & 200   & 0.00  & 0.14 \\[0.05cm]
				& all   & 70 & 23.72 & 811.29 &       & 400 & 0.00 & 0.17 \\
				\midrule
				\multirow{3}[2]{*}{medium} & low   & 68    & 31.91 & 813.37 &       & 200   & 0.00  & 0.19 \\
				& high  & 26    & 51.79 & 1055.41 &       & 200   & 0.00  & 0.13 \\[0.05cm]
				& all   & 94 & 25.42 & 674.19 &       & 400 & 0.00 & 0.16 \\
				\midrule
				\multirow{3}[2]{*}{large} & low   & 68    & 36.69 & 809.30 &       & 200   & 0.00  & 0.19 \\
				& high  & 35    & 51.55 & 1001.10 &       & 200   & 0.00  & 0.13 \\[0.05cm]
				& all   & 103 & 24.07 & 618.65 &       & 400 & 0.00 & 0.16 \\
				\midrule
				\multicolumn{2}{c}{} & {\bf 267} & {\bf 24.41} & {\bf \phantom{0}701.37} &       & {\bf 1200} & {\bf 0.00} & {\bf 0.16} \\
				\bottomrule				
			\end{tabular}%
	}
{}
\end{table}%


\section{Conclusion and Future Research}
\label{section:conclusion}

In this study, we presented an exact cutting plane algorithm for the selective graph coloring problem in perfect graphs, which is a generalization of the method presented in \citep{seker2019decomposition}. 
We also introduced an algorithm to generate random perfect graphs, which, to the best of our knowledge is the first algorithm for this purpose in the literature. 
Given an input graph together with vertex set partition, we search for and optimal selection in the master problem of our cutting plane procedure, and in the subproblem, we find a maximum clique in the graph induced by the selected set of vertices by making use of two alternative methods. 
One of these makes use of semidefinite programming models and works in polynomial time in theory in the class of perfect graphs. 
The other one is a general-purpose maximum clique algorithm from the literature and performs quite efficiently in practice. 
We conducted an extensive computational study on a large set of test instances we generated randomly, in order to test the performance of our algorithm in comparison to those of an IP formulation and a branch-and-price algorithm from the literature.
The computational results show that the cutting plane algorithm significantly improved the solvability of the problem in general, but most evidently in low-density graph instances. 
We also compared the performance of our cutting plane algorithm for perfect graphs in its general form to that of our previous algorithm tailored for three subclasses of perfect graphs; namely, permutation, generalized split, and chordal graphs. The use of our cutting plane algorithm for general perfect graphs resulted in better performance in permutation graphs, and marked deterioration in chordal graphs regardless of edge density. In the class of generalized split graphs, the overall performance became worse with the algorithm for general perfect graphs.

Our presented solution strategy can be adapted to general graphs where the clique number is not necessarily equal to the chromatic number. 
We carried out preliminary experiments to test the performance of such an adaptation (see Appendix). 
In order to apply this exact solution approach to general graphs, a minimum coloring algorithm is required, while it is still possible and indeed helpful to utilize a maximum clique algorithm as well.
Since the maximum clique algorithm by \cite{tomita2010simple} works for any graph, we keep utilizing it to generate constraint \eqref{cliquecut}. 
In order to generate constraint \eqref{cut1}, we employed two different minimum coloring algorithms in our subproblem, which are two best-performing ones according to a survey on vertex coloring problems by \cite{malaguti2010survey}:
(i) a branch-and-cut algorithm by \cite{mendez2006branch}, and (ii) a column generation algorithm by \cite{mehrotra1996column}, for which we used the implementation by \cite{held2012maximum}. 
Unfortunately, neither (i) nor (ii) performs very well, such that the majority of the solution time is spent in the subproblem of our cutting plane method.
The reason for the coloring task to constitute the bottleneck of our cutting plane approach on general graphs is not only that the coloring problem turns out to be much harder in practice as compared to the maximum clique problem, but also that the weakness of constraint \eqref{cut1} induces more calls to the coloring problem. 
Hence, further research should investigate how to design and incorporate alternative stronger cuts to improve the performance of our cutting plane method. Moreover, we note that our approach of taking advantage of the graph structures can be potentially applied to other exact methods to solve {\sc Sel-Col}.

As for perfect graph generation, our proposed algorithm can be enriched by including different perfection-preserving methods. 
One can also address the more challenging open question of generating every perfect graph with strictly positive probability, or at least being able to give a measure of the quality or the distribution of the generated perfect graphs \citep{roussel2009strong}.

\section*{Acknowledgements}
We are grateful to Pinar Heggernes for her helpful suggestions in the perfect graph generation algorithm when she was on her sabbatical leave at Boğaziçi University and when the first author was a visiting scholar at the University of Bergen. We are also thankful to three anonymous referees, whose constructive comments and useful suggestions helped us improve the content and exposition of the paper. 

\bibliographystyle{elsarticle-harv}
\bibliography{references}

\newpage
\section*{Appendix: Extension to General Graphs}

As we noted in Section \ref{section:conclusion}, the presented cutting plane algorithm can be adapted as an exact method to solve {\sc Sel-Col} in general graphs, i.e., graphs with no particular structure, where the chromatic number is not necessarily equal to the clique number. 
This can be done by only generating constraint \eqref{cut1} in the subproblem.
However, since we observed that constraint \eqref{cliquecut} can be produced efficiently and facilitate the solution procedure considerably, we add constraints of type \eqref{cliquecut} as long as they are violated, and start adding \eqref{cut1} only afterwards, i.e., when no violated constraint \eqref{cliquecut} is left.

The maximum clique algorithm by \cite{tomita2010simple} works for any graph.
Therefore, we keep utilizing it for general graphs to generate constraint \eqref{cliquecut}.
In order to generate constraint \eqref{cut1} for general graphs, we need to employ a minimum coloring algorithm in our subproblem.
According to a survey on vertex coloring problems by \cite{malaguti2010survey}, two best performing exact algorithms for graph coloring are a column generation approach by \cite{mehrotra1996column} and a branch-and-cut algorithm by \cite{mendez2006branch}. 
We used our own implementation of the method by \cite{mendez2006branch} as the first alternative, and an implementation of the column generation approach of  \cite{mehrotra1996column} offered by \cite{held2012maximum} as the second one, which is referred to as {\it exactcolors}. 

In the sequel, we present a preliminary set of computational results we obtained using the extension of our cutting plane approach for graphs with no particular structure. 
Using the two alternative methods to solve the minimum coloring problem in our subroutine, we conducted experiments on a set of 24 \cite{erdos1959random} type random graphs having 50 to 500 vertices and four different average edge densities, which are 0.1, 0.3, 0.5, and 0.7.  
For each pair of $n$ and edge density value, we used a single graph instance and coupled it with a set of clusters, each containing 2 to 5 vertices. 
We set a time limit of 1200 seconds in each run, as before. 

Table \ref{tab:generalgraphs} provides a summary of our results for general graphs, with a structure similar to the previous ones. The two sets of columns with headings ``Cutting plane with B\&C" and ``Cutting plane with {\it exactcolors}" stand for the cutting plane algorithm coupled with the minimum coloring algorithms by  \cite{mendez2006branch} and \cite{held2012maximum}, respectively. Columns ``\% time in cliq" and ``\% time in col" respectively give the percentage of solution time spent in solving the maximum clique and minimum coloring problems in the subproblem. 
\begin{table}[htbp]
  \caption{Summary of experimental results for general graph instances}
  \label{tab:generalgraphs}
    \centering
  	\resizebox{0.99\textwidth}{!}{
    \begin{tabular}{c S[table-format=3.1] 
    c S[table-format=2.2] S[table-format=4.2]
    c S[table-format=3.2]
    c S[table-format=2.2] S[table-format=1.2] S[table-format=2.2] S[table-format=4.2]
    c S[table-format=2.2] S[table-format=1.2] S[table-format=2.2] S[table-format=4.2] }
    \toprule
          &       & \multicolumn{3}{c}{\bf IP formulation} & \multicolumn{2}{c}{\bf B\&P} & \multicolumn{5}{c}{\bf Cutting plane with B\&C} & \multicolumn{5}{c}{\bf Cutting plane with \emph{exactcolors}} \\
          \cmidrule(lr){3-5}\cmidrule(lr){6-7}\cmidrule(lr){8-12}\cmidrule(lr){13-17}
			\parbox[t]{1cm}{\centering $\boldsymbol{n}$} & \parbox[t]{1cm}{\centering \bf Avg \# clust} & \parbox[t]{1cm}{\centering \bf \# opt} & \parbox[t]{1cm}{\centering \bf \% gap} & \parbox[t]{1.2cm}{\centering \bf Avg time} & \parbox[t]{1cm}{\centering \bf \# opt} & \parbox[t]{1.2cm}{\centering \bf Avg time} & \parbox[t]{1cm}{\centering \bf \# opt} & \parbox[t]{1cm}{\centering \bf \% gap} & \parbox[t]{1.2cm}{\centering \bf \% time in cliq} & \parbox[t]{1cm}{\centering \bf \% time in col } & \parbox[t]{1.2cm}{\centering \bf Avg time} & \parbox[t]{1cm}{\centering \bf \# opt} & \parbox[t]{1cm}{\centering \bf \% gap} & \parbox[t]{1.2cm}{\centering \bf \% time in cliq} & \parbox[t]{1cm}{\centering \bf \% time in col} & \parbox[t]{1.2cm}{\centering \bf Avg time} \\
			\midrule
            50    & 14.5 & 4     & 0.00  & 2.45  & 4     & 1.84  & 4     & 0.00  & 0.25  & 71.69 & 12.60 & 4     & 0.00  & 0.80  & 16.52 & 1.62 \\[0.05cm]
            100   & 28.8 & 3     & 4.17  & 478.60 & 4     & 15.26 & 2     & 14.58 & 0.28  & 82.89 & 697.49 & 2     & 15.71 & 0.58  & 28.22 & 629.53 \\[0.05cm]
            200   & 57.8 & 1     & 52.59 & 940.58 & 4     & 123.32 & 1     & 49.62 & 0.01  & 99.89 & 1197.50 & 0     & 67.06 & 0.10  & 92.08 & 1200.00 \\[0.05cm]
            300   & 84.0 & 0     & 87.50 & 1200.00 & 4     & 510.17 & 0     & 65.00 & 0.01  & 99.95 & 1200.00 & 0     & 75.99 & 0.02  & 99.37 & 1200.00 \\[0.05cm]
            400   & 114.5 & 0     & 99.78 & 1200.00 & 3     & 904.82 & 0     & 76.85 & 0.01  & 99.97 & 1200.00 & 0     & 88.43 & 0.00  & 97.53 & 1200.00 \\[0.05cm]
            500   & 144.8 & 0     & 99.81 & 1200.00 & 1     & 1172.96 & 0     & 76.53 & 0.05  & 99.94 & 1200.00 & 0     & 90.92 & 0.02  & 88.90 & 1200.00 \\
            \midrule
          &       & {\bf 8} & {\bf 57.31} & {\bf \phantom{0}836.94} & {\bf 20} & {\bf 454.73} & {\bf 7} & {\bf 47.10} & {\bf 0.10} & {\bf 92.39} & {\bf \phantom{0}917.93} & {\bf 6} & {\bf 56.35} & {\bf 0.25} & {\bf 70.44} & {\bf \phantom{0}905.19} \\
          \bottomrule
    \end{tabular}%
}%
{}%
\end{table}%

The overall summary of the results delivered in the last row of Table \ref{tab:generalgraphs} indicates that B\&P method outperforms all of the others by optimally solving 83.3\% of the instances in approximately half the time the cutting plane method takes on the average. 
The fraction of time our cutting plane method spends in the minimum coloring problem is nearly one in many cases, implying that the coloring task constitutes the bottleneck of our algorithm. The reason for this is not only that the coloring algorithms take too much time; the constraint \eqref{cut1} being weak induces more calls to the coloring problem, hence increasing the total amount of time spent for it. So, it is the coexistence of the weakness of constraint \eqref{cut1} and the elevated solution times spent for the coloring problem that leads to unsatisfactory performance of the cutting plane algorithm in general graph instances.


\end{document}